\newif\ifpnote\pnotetrue{}
\newif\ifdraft\draftfalse{}
\newif\ifcolor\colortrue{}
\newif\ifanon\anonfalse{}
\newif\iftr\trtrue{}

\documentclass[runningheads]{llncs}

\InputIfFileExists{pervasives.sty}{}{\typeout{No file pervasives.sty}}

\usepackage{microtype}
\usepackage{graphicx}

\makeatletter
\renewcommand{\ALG@beginalgorithmic}{\ttfamily}
\makeatother

\begin{document}
\title{\TXF: A DSL for Concurrent Filestores}


\ifanon
\author{Paper \#47}
\institute{\vspace*{-1em}}
\else
\author{
Jonathan DiLorenzo\inst{1} \and
Katie Mancini\inst{1} \and
Kathleen Fisher\inst{2} \and
Nate Foster\inst{1}}

\authorrunning{J. DiLorenzo et al.}

\institute{
Cornell University, Ithaca NY 14850, USA \and
Tufts University, Medford MA 02155, USA
}
\fi

\maketitle
\setcounter{footnote}{0}
\begin{abstract}

Many systems use ad hoc collections of files and directories to store
persistent data. For consumers of this data, the process of properly
parsing, using, and updating these \emph{filestores} using
conventional APIs is cumbersome and error-prone. Making matters worse,
most filestores are too big to fit in memory, so applications must
process the data incrementally while managing concurrent accesses by
multiple users. This paper presents \TXFFull{} (\TXF{}), which builds
on earlier work on Forest to provide a simpler, more powerful API for
managing filestores, including a mechanism for managing concurrent
accesses using serializable transactions. Under the hood, \TXF{}
implements an optimistic concurrency control scheme using
Huet's \emph{zippers} to track the data associated with filestores. We
formalize \TXF{} in a core calculus, develop a proof of
serializability, and describe our OCaml prototype, which we have used
to build several practical applications.



\end{abstract}

\section{Introduction}\label{sec:intro}
%
%
%
\pnote{Ad hoc data is everywhere, but problematic.}
Modern database systems offer numerous benefits to programmers,
including rich query languages and impressive performance. However,
programmers in many areas including finance, telecommunications, and
the sciences, rely on \textit{ad hoc} data formats to store persistent
data---e.g., flat files organized into structured directories. This
approach avoids some of the initial costs of using a database such as
writing schemas, creating user accounts, and importing data, but it
also means that programmers must build custom tools for correctly
processing the data---a cumbersome and error-prone task.

%
%
%
\pnote{Concurrency issues are extra difficult and in a class of their own.}
Applications often have an additional class of critical errors arising from
concurrency. Frequently, applications that store large amounts of persistent
data in the file system have multiple users that may be reading and writing the
data concurrently, or single users relying on parallelism to speed up their
work. For example, many instructors in large computer science courses use
filestores to manage student data, using ad hoc collections of assignment
directories, grade rosters stored in CSV files, and grading scores and comments
stored in ASCII files. To automate common grading tasks---e.g., computing
statistics, normalizing raw scores, uploading grades to the
registrar---instructors often write scripts to manipulate the data. However,
these scripts are written against low-level file system APIs, and rarely handle
multiple concurrent users. This can easily lead to incorrect results or even
data corruption in courses that rely on large numbers of TAs to help with
grading.

%
%
%
\pnote{PADS and Forest mitigate the standard issues, but not concurrency.}
The PADS/Forest family of languages offers a promising approach for managing ad
hoc data. With these languages, the programmer specifies the structure of an ad
hoc data format using a simple, declarative specification, and the compiler
generates an in-memory representation for the data, load and store functions for
mapping between in-memory and on-disk representations, as well as tools for
analyzing, transforming, and visualizing the data. PADS focused on ad hoc data
stored in individual files~\cite{fisher+:pads}, while Forest handles ad hoc data
in \textit{filestores}---i.e., structured collections of files, directories, and
links~\cite{forest-icfp:fisher+}. Unfortunately, the languages that have been
proposed to date lack support for concurrency.

%
%
\pnote{We introduce TxForest which provably guarantees serializability.}
To address this challenge, this paper proposes \TXFFull{} (\TXF{}),
a declarative domain-specific language for correctly processing ad hoc
data in the presence of concurrency. Like its predecessors, \TXF{}
uses a type-based abstraction to specify the structure of the data and
its invariants. From a \TXF{} description, the compiler generates a
typed representation of the data as well as a high-level programming
interface that abstracts away direct interactions with the file system
and provides operations for automatically loading and storing data
from the underlying file system, while gracefully handling errors. In
addition, \TXF{} guarantees serializable semantics for transactions.

\jdl{TODO2: Consider revising and combining Zipper paragraphs (perhaps putting
properties before description?) and including more or different information
about why they are a good fit. See the second paragraph of Section 3 (soon)}

%
%
\pnote{TxForest uses zippers as an underlying abstraction and the secret sauce}
The abstraction that facilitates the serializable semantics, along with a slew
of additional desired properties are \emph{zippers}. \TXF{} uses a
tree-structured representation based on Huet's
Zippers~\cite{Huet:1997:ZIP:969867.969872} to represent the filestore being
processed. Rather than representing a tree in terms of the root node and its
children, a zipper encodes the current node, the path it traversed to get there
and the nodes that it encountered along the way. Importantly, local changes to
the current node as well as many common navigation operations involving adjacent
nodes can be implemented in constant time. For example, by replacing the current
node with a new value and then `zipping' the tree back up to the root,
modifications can be implemented in a purely-functional way.

%
%
%
\pnote{Why are zippers the right thing?}
As others have also observed~~\cite{OlegZFSTalk}, zippers are a good
abstraction for filestores, for several reasons: (1) The concept of
the working path is cleanly captured by the current node; (2) Most
operations are applied close to the current working path; (3) The
zipper naturally captures incrementality by loading data as it
is encountered in the zipper traversal; (4) A traversal (along with
annotations about possible modification) provides all of the
information necessary to provide rich semantics, such as
copy-on-write, as well as a simple optimistic concurrency control
scheme that guarantees serializability.

\jdl{Reviewer 3 quibbled about capitalizing semi-colon separated lists. I think
I prefer it this way, but I'm not sure if one way is more correct. Since I think
Nate at least edited this list quite a bit, I'm guessing this is correct?}

%
%
%
\pnote{Brief Paper Summary.}
In this paper, we formalize the syntax and semantics of \TXF{} for a
single thread of execution, and we establish various correctness
properties, including roundtripping laws in the style of
lenses~\cite{focal-toplas}. Next, we extend the semantics to handle
multiple concurrent threads of execution, and introduce a transaction
manager that implements a standard optimistic concurrency scheme. We
prove that all transactions that sucessfully commit are serializable
with respect to one another. Finally, we present a prototype
implementation of \TXF{} as an embedded language in OCaml,
illustrating the feasibility of the design, and use it to implement
several realistic applications.

%
%
\pnote{Contributions.}
%
Overall, the contributions of this paper are as follows:
\begin{itemize}
\item We present \TXFFull{}, a declarative domain-specific language for
processing ad hoc data in concurrent settings (\secref{sec:forest,sec:concurrency}).
\item We describe a prototype implementation of \TXFFull{} as an embedded
domain-specific language in OCaml (\secref{sec:implementation}).
\item We prove formal properties about our design including serializability and
round-tripping laws.
\end{itemize}

%
%
\pnote{Overview of the sections to come.}
The rest of this paper is structured as follows: \secref{sec:overview}
introduces a simple example to motivate \TXF{}. \secref{sec:forest}
presents the syntax and single-threaded semantics
of \TXF{}. \secref{sec:concurrency} adds the multi-threaded semantics
and the serializability theorem. \secref{sec:implementation} discusses
the OCaml implementation of \TXF{} and an application. We review related
work in \secref{sec:related} and conclude in \secref{sec:conclusion}.
The proofs of formal properties are in \appOrTr{}.

\section{Example: Course Management System}\label{sec:overview}

%
%
%
\pnote{Section Overview: Course management system motivating the design.}
This section introduces an example of an idealized course management system to
motivate the design of \TXF{}. \figref{fig:grading} shows a fragment of a
filestore used in tracking student grades. The top-level directory
(\texttt{grades}) contains a set of sub-directories, one for each homework
assignment (\texttt{hw1}--\texttt{hw5}). Each assignment directory has a file
for each student containing their grade on the assignment (e.g.,
\texttt{aaa17}), as well as a special file (\texttt{max}) containing the maximum
score for that homework. Although this structure is simple, it closely resembles
filestores that have actually been used to keep track of grades at several
universities.

\begin{figure}[t]
  \centerline{\includegraphics[scale=.55]{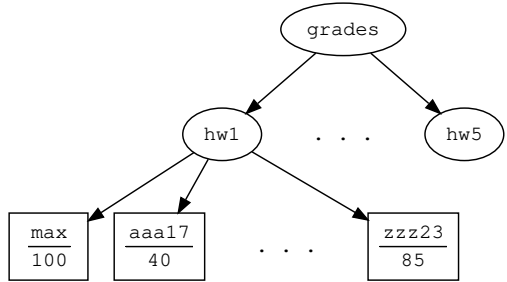}}
  \caption{Example: file system fragment used to store course data.}
  \label{fig:grading}
\end{figure}

%
%
%
\pnote{Example description with normalization and concurrency issues.}
There are various operations that one might want to perform on this
filestore, but to illustrate the challenges related to concurrency, we
will focus on normalization. Normalization might be used to ensure
that the grades for a particular homework fall between some specified
limits or match a given probability distribution. We assume an idempotent
normalization operation \texttt{f} that receives various assignment
statistics and the current score and computes a normalized score.

\paragraph*{OCaml Implementation.}
%
%
%
\pnote{OCaml Implementation intro.}
To start, let us see how we might write a renormalization procedure for this
filestore in a general-purpose language---e.g., OCaml. For simplicity, the code
relies on helper functions, which are explained below.
%
%
\begin{code}
let renormalize f hw gmin =
  let hwDir = sprintf "grades/hw
  let gmax = get_score (hwDir ^/ "max") in
  let studentFiles = get_students hwDir in
  let (cmin, cmax) = get_min_and_max studentFiles in
  map_scores (f cmin cmax gmin gmax) studentFiles
\end{code}
%
%
%
\pnote{Explanation of code.}
The \texttt{renormalize} function takes as input the function to use
to normalize individual scores (\texttt{f}), the identifier of a
homework assignment (\texttt{hw}), and the minimum score to use when
scaling scores (\texttt{gmin)}. It retrieves the value from the
\texttt{max} file, using the \texttt{get\_score} helper, which  reads the file and parses it
into a score. Next, it retrieves the list of paths to every
student file (\texttt{studentFiles}). It then computes the minimum
(\texttt{cmin}) and maximum (\texttt{cmax}) score over all students
using a helper function (\texttt{get\_min\_and\_max}), which again
accesses data in the underlying file system. Finally, it maps the
function \texttt{f} over each student's score, together with the
aggregate statistics supplied as arguments, and writes the new score
back to the file, again using a helper function to perform the
necessary iteration (\texttt{map\_scores}) and file writes.

%
%
%
\pnote{Standard Forest Errors.}
Although this procedure is simple, there are a number of potential
pitfalls that could arise due to its use of low-level file system
APIs. For example, one of the files or directories might not exist or
there might be extra files in the file system. The structure of the
filestore might be malformed, or might change over time. Any of these
mistakes could lead to run-time errors or worse, they might silently
succeed, but produce incorrect results.

%
%
%
\pnote{Concurrency issue which is hard to detect, diagnose, and fix.}
This implementation also suffers from a more insiduous set of problems related
to concurrency. Consider what happens if multiple members of the
course staff execute the renormalization procedure concurrently. If
the stage that computes the minimum and maximum scores is interleaved
with the stage that invokes \texttt{f} and writes the normalized
values back to the file system, we could easily be left with a mangled
filestore and incorrect results---something that would likely be
difficult to detect, diagnose, and fix.

\paragraph*{Classic Forest Implementation.}
%
%
%
\pnote{Forest Implementation Intro.}
Next let us consider an implementation in
Forest~\cite{forest-icfp:fisher+}. The programmer starts by explicitly
specifying the structure of their filestore using the following
declarations:
%
%
\begin{code}
  grades = [hw :: hws | hw <- matches RE "hw[0-9]+"]
  students = file
  hws = directory \TextSet{
    max is "max" :: file;
    students is [student :: students | student <- matches RE "[a-z]+[0-9]+"];
  }
\end{code}
%
%
%
\pnote{Explanation of specification and what it generates.}
The \texttt{grades} specification describes the structure of the
top-level directory: a list of homework directories, each containing a
file named \emph{max} and a list of \texttt{students} (each
represented as a \texttt{file}\footnotemark).

\footnotetext{By integrating with PADS~\cite{fisher+:pads}, we could go  a step further and specify
the contents of the file as well---i.e. a single line containing an integer.}

Given this specification, the Forest compiler generates an in-memory
representation for the data, as well as associated functions for loading data
from and storing data to the file system. For example, the types generated from
the \texttt{grades} specification, and the representation of \texttt{hws}, are:
%
%
\begin{code}
  type grades_rep = hws_rep list
  type grades_md = hws_md list md
  type hws_rep = \Set{ max : string; students : students_rep list}
  val grades_load : filepath -> grades_rep * grades_md
  val grades_store : filepath -> grades_rep * grades_md -> unit
\end{code}

%
%
%
\pnote{Explanation of generated code.}
The \texttt{md} types store metadata about the Forest computation,
including permissions and whether any errors were encountered.
The \texttt{load} and \texttt{store} functions map between the on-disk
and in-memory representations, and automatically check for errors and
inconsistencies in the data. Using these functions, we write
the \texttt{renormalize} procedure as follows:

%
%
\begin{code}
  let renormalize f hw gmin : unit =
    let (gr,gmd) = grades_load (baseDir ^/ "grades") in
    if gmd.num_errors = 0 then
      let (hwr,hwmd) = find (sprintf "hw
      let gmax = get_score hwr.max in
      let (cmin, cmax) = get_min_max hwr in
      map_scores (f cmin cmax gmin gmax) hwr hwmd
    else
      failwith (String.concat "\\n" gmd.error_msg)
\end{code}

%
%
%
\pnote{Explanation of code.}
This code is similar to the OCaml implementation, but there are a few
key differences. It first loads the entire \texttt{grades} directory
and checks that it has no errors. This makes the auxilliary functions,
like
\texttt{get\_score} (which now just turns a string into an integer) and
\texttt{set\_score} simpler and more robust, since they no longer
need to worry about such issues. It then locates the representation
and metadata for the assignment, computes aggregate statistics, and
invokes \texttt{f} to renormalize and update the scores.
The \texttt{get\_min\_max} and \texttt{map\_scores} helpers are
similar to the direct versions discussed previously.

%
%
%
\pnote{This solves the standard issues, but fails to solve concurrency.}
The Forest implementation offers several important benefits over the
OCaml code: (1) The structure of the filestore is explicit in the
specification and the code; (2) The use of types makes certain
programming mistakes impossible, such as attempting to read a file at
a missing path; and (3) Any part of the filestore not conforming to
the specification is automatically detected.

However, the Forest code still suffers from the same concurrency
issues discussed above. Further, it is unnecessary (and often
infeasible) to load the entire filestore into memory at once,
particularly when we only want to manipulate one homework, or one
student. While we could have directly loaded a single homework or
student with their associated load functions, we would not get as much
information about possible errors.

\paragraph*{\TXFFull{} Implementation.}
%
%
%
\pnote{TxForest intro.}
\TXF{} offers the same advantages as Forest, while dealing with
issues related to concurrency and incrementality. The only cost is a
small shift in programming style---i.e., navigating using a zipper.

%
%
%
\pnote{How specifications are used in TxForest.}
The \TXF{} specification for our running example is identical to the
Forest version. However, this surface-level specification is then
translated to a core language (\secref{sec:forest}) that uses Huet's
zipper internally and also provides transactional guarantees.
The \TXF{} code for the \texttt{renormalize} function is different
than the Forest version. Here is one possible implementation:

%
%
\begin{code}
  let renormalize f hw gmin zipper : (unit,string) Result.t =
    let
    let
    let
    let
    map_scores (f cmin cmax gmin gmax) studentZ
\end{code}
%
%
%
\pnote{Explanation of code.}
Note that the type of the function has changed so that it takes a
zipper as an argument and returns a value in the result monad:
\begin{code}
  type ('a,'b) Result.t = Ok 'a | Error of 'b
\end{code}

Intuitively, this monad tracks the same sorts of errors seen in the
Forest code---e.g. from malformed filestores, and does not include
concurrency issues.

The \texttt{goto\_name\_p} function traverses the zipper---e.g.,
\texttt{goto\_name\_p "hw1" zipper} navigates to the comprehension
node named \emph{hw1} and then down to the corresponding file system
path, ending up at a \texttt{hws} node. The \emph{bind} operator
(\texttt{>>=}) threads the resulting zipper through the monad. The
\texttt{let\%bind $\VarMeta$ = $\ExpMeta_1$ in $\ExpMeta_2$}
syntax is shorthand for \texttt{$\ExpMeta_1$ >>= fun $\VarMeta$ ->
$\ExpMeta_2$}. The \texttt{goto} function is similar, but is limited
to directories and does not walk down the last path operator. Finally,
the helper functions, \texttt{map\_scores} and \texttt{get\_min\_max},
use TxForest library functions to map and fold over the zipper
respectively.

\pnote{Introduction of transaction running constructs.}
To use the \texttt{renormalize} function, users need some way to
construct a zipper. The TxForest library provides functions
called \texttt{run\_txn} and \texttt{loop\_txn}:
%
%
\begin{code}
  type txError = TxError | OpError of String
  val run_txn : spec -> path -> (t -> ('a,string) Result.t)
                -> (unit -> ('a,txError) Result.t)
  val loop_txn : spec -> path -> (t -> ('a,string) Result.t)
                -> (unit -> ('a,string) Result.t)
\end{code}
which might be used as follows:
%
%
\begin{code}
  match run_txn grades_spec "grades" (renormalize 1 60) () with
  | Error TxError -> printf "Transaction aborted due to conflict"
  | Error (OpError err) -> printf "Transaction aborted due to error: 
  | Ok _ -> printf "Renormalization successful"
\end{code}

The \texttt{run\_txn} function takes a specification, an initial path,
and a function from zippers (\texttt{t}) to results and produces a
thunk. When the thunk is forced, it generates a zipper focused on the
given path and runs the function. If this execution results in an
error, the outer computation produces an \texttt{OpError}. Otherwise,
it attempts to commit the modifications produced during the
computation. If this succeeds, it returns the result of the function,
otherwise it discards the results and returns a \texttt{TxError}.
The \texttt{loop\_txn} function is similar, but retries the
transaction until there is no conflict or the input function produces
an error.

%
%
%
\pnote{Concurrency issues go away! (And incrementality?).}
\TXF{} guarantees that transactions will be serializable with respect
to other transactions---i.e., the final file system will be equivalent
to one produced by executing the committed transactions in some
serial order. See \secref{sec:concurrency} for the formal concurrent
semantics and the serializability theorem. In our example, this means
that no errors can occur due to running multiple renormalization
transactions simultaneously. Furthermore, \TXF{} automatically
provides incrementality by only loading the data needed to traverse
the zipper---an important property in larger filestores. Incremental
Forest~\cite{dilorenzo2016incremental} provides a similar facility, but
requires explicit user annotations.

\jdl{Reviewer 3 notes that incremental Forest is mentioned here in passing and
suggests that moving Related Work to beginning of paper might help avoid reader
confusion. I don't think moving the Related Work is good and given that we talk
about Classic Forest earlier in this section, I'm unconvinced that this mention
is a problem.}

%
%
%
\pnote{Segue to formalism.}
Overall, \TXF{} provides incremental support for filestore
applications in the presence of concurrency. The next two sections
present the language in detail, develop an operational model, and
establish its main properties.

\jdl{Why operational model? It's mostly denotational? Or is that not what we
mean by this...?}

\section{\TXFFull}\label{sec:forest}
%
%
%
\pnote{Section Overview: Tie back to main idea (need formalism for proofs?),
local semantics, roadmap (multiple threads will come up next).}
This section presents \TXF{} in terms of a core calculus. We discuss the goals
and high level design decisions of our language before formalizing the syntax
and semantics as well as several properties including round-tripping laws for
fetching and storing data, equational identities, and filestore consistency
relations. This section deals primarily with the single-threaded semantics,
while the next section presents a concurrent model.

\pnote{Language requirements.}
The goals of this language are to allow practical processing of filestores for
non-expert users. This leads to several requirements: (1) An intuitive way of
specifying filestores, which has been solved in previous
work~\cite{forest-icfp:fisher+}; (2) Automatic incremental processing, as
filestores are often too large to fit in memory; (3) Automatic concurrency
control, since concurrency is both common and difficult to get right; and (4)
Transparency: since filestore interaction is often expensive, it should be
explicit.

\jdl{Not clear that Transparency is useful here or that cool.}

\pnote{Zipper and Language interdesign.}
The zipper abstraction that our language is based on helps us achieve our second
and fourth requirement. Both of these requirements, and concurrency are then
further addressed by our locality-centered language design. The semantics of
every command and expression is designed to only consider the locale around the
focus node of the zipper. This means that every command only needs to look at a
small part of the filestore, which, along with the fact that data can be loaded
only as-required while traversing the zipper, gives us incrementality. We
believe that this locality and the explicit zipper traversal commands also lends
us transparency. In particular, the footprint of any command is largely
predictable based on the filestore specification and current state. Finally,
this predictability makes concurrency control simpler by making logging an easy
affair.

\subsection{Syntax}\label{subsec:forest-syntax}

\begin{figure}[t]
  \[
  \begin{array}{@{}l@{\;}r@{\;}cl@{}@{~~}r}
  \text{Strings} & \StrMeta{}  & \in & \StrSet{} &  \\
  \text{Paths} & \PathMeta{} & \in & \PathSet{} &   \\
  \text{Integers} & \IntMeta{} & \in & \IntSet{} &   \\
  \text{Variables} & \VarMeta{} & \in & \VarSet{} &  \\
  \text{Values} & \ValMeta{} & \in & \ValSet{} &   \\
  \text{Environments } & \EnvMeta \in \EnvSet & : &\map{\VarSet}{\ValSet} \\
  \text{File Systems} & \FSMeta  & : & \map{\PathSet}{\ContentsSet}  \\
  \text{Programs} & \ProgMeta  & \bnfdef & (\PathMeta,\SpecMeta,\CommandMeta)
  \\
  \text{Contents} & \ContentsMeta  & \bnfdef & \FSDirDef \alt \FSFileD
  \\
  \text{Commands} & \CommandMeta  & \bnfdef & \FCommandMeta \alt \ComSkip \alt \ComSeqD \alt \ComAssignD \\
  && \alt & \ComIfD \alt \ComWhileD
  \\
  \text{Expressions} & \ExpMeta, \BoolMeta  & \bnfdef & \FExpMeta
    \alt \ValMeta
    \alt \VarMeta \alt \ExpAppD \alt \dots
  \\
  \text{Local Contexts} & \ContextMeta  & : & \product{\EnvSet}{\product{\PathSet}{\ZipperSet}}
  \\
  \text{Global Contexts} & \GContextMeta & : & \product{\PathSetSet}{\FSSet}
  \end{array}
  \]
  \caption{Preliminaries}
  \jdl{TODO4: Perhaps just have a context and a separate FS}
  \jdl{We should probably add a simple path syntax here}
  \jdl{Reviewer 3 would like us to explain $\Set{\overline{\StrMeta}}$. I think
  it's pretty self-explanatory...?}
  \label{fig:forest-meta}
\end{figure}

%
%
%
\pnote{FS Model.}
In our formal model, we view a file system as a map from paths to file system
contents, which are either directories (containing their children's names) or
files (containing strings). For a path and file system, $\PathMeta$ and
$\FSMeta$, we define $\PathMeta \in \FSMeta \eqdef \PathMeta \in \dom{\FSMeta}$.
See \figref{fig:forest-meta} for the metavariable conventions used in our
formalization.

We will assume that all file systems are well formed---i.e., that they
encode a tree, where each node is either a directory or a file with no
children:

\begin{definition}[Well-Formedness]\label{def:fs-wf}
  A file system $\FSMeta$ is well-formed iff:
  \begin{enumerate}
    \item $\MapGet{\FSMeta}{\PathRoot} = \FSDirF{\ignoreArg}$ (where $\PathRoot$ is the root node)
    \item $\PathCons{\PathMeta}{\StrMeta} \in \FSMeta \iff
          \MapGet{\FSMeta}{\PathMeta} = \FSDirF{\Set{\StrMeta; \dots}}$
  \end{enumerate}
\end{definition}

%
%
\begin{figure}[t]

  \[
  \begin{array}{lrcl}
  \text{Specifications} & \SpecMeta \in \SpecType & \bnfdef &
  \FSpecFile \alt \FSpecDir
  \alt \FSpecPathD \alt \FSpecPairD
  \alt \FSpecCompD \alt \FSpecOptD \alt \FSpecPredD
  \\
  \text{Zippers} & \ZipperMeta  & \bnfdef & \{
  \NamedType{\ZipParent}{\ZipperSet~\OptType}; \\
  &&& ~~ \NamedType{\ZipLeft}{(\FZipNodeType)~\ListType}; \\
  &&& ~~ \NamedType{\ZipCurrent}{(\FZipNodeType)}; \\
  &&& ~~ \NamedType{\ZipRight}{(\FZipNodeType)~\ListType}
  \}
  \\ \\
  \text{Forest Commands} & \FCommandMeta  & \bnfdef & \FNavigationMeta \alt \FUpdateMeta
  \\
  \text{Forest Navigations} & \FNavigationMeta  & \bnfdef &
    \FComDown
    \alt \FComUp
    \alt \FComNext
    \alt \FComPrev
  \\ && \alt &
    \FComIntoPair
    \alt \FComIntoComp
    \alt \FComIntoOpt
    \alt \FComOut
  \\
  \text{Forest Updates} & \FUpdateMeta  & \bnfdef &
    \FComStoreFileD \alt \FComStoreDirD \alt \FComCreatePath
  \\
  \\
  \text{Forest Expressions} & \FExpMeta  & \bnfdef &
    \FExpFetchFile
    \alt \FExpFetchDir
    \alt \FExpFetchPath
  \\ && \alt &
    \FExpFetchComp
    \alt \FExpFetchOpt
    \alt \FExpFetchPred
  \\ && \alt &
  \FExpRunCD
  \alt \FExpRunED
  \alt \FFunVerify
  \\
  \\
  \text{Log Entries} & \LogEntryMeta & \bnfdef &
  \TxLReadD
  \alt \TxLWriteFileD
  \alt \TxLWriteDirD
  \\
  \text{Logs} & \LogMeta & : & \LogEntrySet ~ \ListType
  \end{array}
  \]
  \jdl{TODO1: Add Paths}
  \caption{Main Syntax}
  \label{fig:forest-syntax}
\end{figure}

%
\jdl{I'm not sure this is where I want this signposting paragraph}
\pnote{Transition/signposting paragraph with an example relating to previous section.}
In the previous section, we gave a flavor of the specifications one might write
in \TXF{}. These were written in our surface language, which compiles down to
the core calculus, whose syntax is given in \figref{fig:forest-syntax}. The core
specifications are described fully below, but first, we will provide the
translation of the \texttt{hws} specification from \secref{sec:overview} to
provide an intuition:
\begin{code}
directory \TextSet{
    max is "max" :: file;
    students is [student::students | student <- matches RE "[a-z]+[0-9]+"]}
\end{code}
becomes
$\FSpecPairF{max}{\FSpecPathF{"max"}{\FSpecFile}}{\FSpecPairF{dir}{\FSpecDir}{
  \FSpecCompF{\FSpecPathF{student}{students}}{student}{\ExpMeta}}}$
where $\ExpMeta = \texttt{filter }(\FExpRunF{\FExpFetchDir}{dir}) \texttt{ "[a-z]+[0-9]+"}$.

Directories become dependent pairs, allowing earlier parts of directories to be
referenced further down. Comprehensions which use regular expressions to query the
file system, also turn into dependent pairs. The first part is a
$\FSpecDir{}$, which, in the second part, is fetched and filtered using the regular expression.

\pnote{Syntax description: Specifications.}
Formally, \TXF{} specifications $\SpecMeta$ describe the
shape and contents of \emph{filestores}, structured subtrees of a file system.
They are almost identical to those in Classic Forest~\cite{forest-icfp:fisher+}
and, to a first approximation, can be understood as follows:
\begin{itemize}
  \item \emph{Files and Directories.} The $\FSpecFile{}$ and $\FSpecDir{}$
  specifications describe filestores with a file and directory, respectively, at
  the current path.
  \item \emph{Paths.} The $\FSpecPathD$ specification describes a filestore modeled by $\SpecMeta$ at
  the extension of the current path by the evaluation of $\ExpMeta$.
  \item \emph{Dependent Pairs.} The $\FSpecPairD$ specification describes a
  filestore modeled by both $\SpecMeta_1$ and $\SpecMeta_2$.
  Additionally, $\SpecMeta_2$ may refer to a context constructed from
  $\SpecMeta_1$ through the variable $\VarMeta$.
  \item \emph{Comprehensions.} The $\FSpecCompD$ specification describes a
  filestore modeled by $\SpecMeta$ when $\VarMeta$ is bound to any
  element in the evaluation of $\ExpMeta$.
  \item \emph{Options.} The $\FSpecOptD$ specification describes a filestore
  that is either modeled by $\SpecMeta$ or where the current path does not
  exist.
  \item \emph{Predicates.} The $\FSpecPredD$ specification describes a filestore
  where $\ExpMeta$ evaluates to the boolean $\btrue$. This
  construct is usually used with dependent pairs.
\end{itemize}

\jdl{Reviewer 3 has trouble understanding this list and in particular paths.
They also don't understand the next paragraph or what 'evaluation of e' means.}

Most specifications can be thought of as trees with as many children as they
have sub-specifications. Comprehensions are the single exception; we think of
them as having as many children as there are elements in the evaluation of
$\ExpMeta$.

%
%
%
\pnote{Syntax description: Zipper.}
To enable incremental and transactional manipulation of data contained
in filestores, \TXF{} uses a \emph{zipper} which is constructed from
a specification. The zipper traverses the specification tree while
keeping track of an environment (for binding variables from dependent
pairs and comprehensions). The zipper can be thought of as
representing a tree along with a particular node of the tree that is
in focus.
\ZipCurrent{} represents the focus node, while
\ZipLeft{} and \ZipRight{} represent its siblings to the left and right
respectively.
\ZipParent{} tracks the focus node's ancestors, by containing the
zipper we came from before moving down to this depth of the tree. Some
key principles to keep in mind regarding the zipper are (1) that the
tree can be unfolded as it is traversed and (2) that operations near
the current node in focus are fast, thus optimizing for locality.

%
%
%
\pnote{Syntax description: Command and Expression.}
To express navigation on the zipper, we use standard imperative (IMP)
commands, $\CommandMeta$, as well as special-purpose Forest Commands,
$\FCommandMeta$, which are divided into Forest
Navigations, \FNavigationMeta{}, and Forest Updates, \FUpdateMeta{}.
Navigation commands are those that traverse the zipper, while Update
commands modify the file system. Expressions are mostly standard and
pure: They never modify the file system and only Forest Expressions
query it. Forest Commands and Expressions will be described in greater
detail in \secref{subsec:forest-semantics}.

%
%
%
\pnote{Syntax description: Log.}
To ensure serializability among multiple \TXF{} threads executing concurrently,
we will use a log composed of a list of entries, \LogEntryMeta{}. \TxLReadD{}
indicates that we have read $\ContentsMeta$ at path $\PathMeta$. \TxLWriteFileD{}
indicates that we have written the file $\ContentsMeta_2$ to path $\PathMeta$,
where $\ContentsMeta_1$ was before. \TxLWriteDirD{} indicates that we have written
the directory $\ContentsMeta_2$ to path $\PathMeta$, where $\ContentsMeta_1$ was
before.

\subsection{Semantics}\label{subsec:forest-semantics}

\jdl{We should define $[-]_c$ as the semantics of commands and give its type}
%
%
%
\pnote{Semantics intro + IMP is standard.}
Having defined the syntax, we now present the denotational semantics
of \TXF{}. The semantics of IMP commands are standard and thus elided.
We start by defining the semantics of a program:
\begin{gather*}
  \FPDenD ~ \eqdef ~
    \texttt{project\_fs} ~ (\FCDenFG{\CommandMeta}{(\SetEmpty,\PathMeta, \MkZip{\FZipNodeSNF{\SetEmpty}{\SpecMeta}}{}{}{})}{\pair{\SetEmpty}{\FSMeta}})
\end{gather*}
%
%
%
\pnote{Semantics Description: Programs + Notation intro.}
The denotation of a \TXF{} program is a function on file systems. We use the
specification $\SpecMeta$, to construct a new zipper, seen in the figure using
our zipper notation defined after this paragraph. Then we construct a new local
context using the zipper and the path $\PathMeta$. Finally, we construct a
global context from the file system $\FSMeta$, execute the command
$\CommandMeta$, and project out the resulting file system.

\jdl{Reviewer 3 would like us to state the type of \texttt{project\_fs}, which I
think is quite clear, and $[-]_c$, which is actually an excellent point, since we
haven't talked about it yet here...}

\begin{definition}[Zipper Notation]\label{def:notation}
  We define notation for constructing and deconstructing zippers. To
  construct a zipper, we write,
  \\
  \(
    \MkZip{\ZipCurrent}{\ZipperMeta}{\ZipLeft}{\ZipRight} \eqdef
    \Set{\ZipParent = \OptSomeF{\ZipperMeta}; \ZipLeft; \ZipCurrent; \ZipRight},
  \)
  where any of \ZipParent{}, \ZipLeft{}, and \ZipRight{} can be left out to denote a
  zipper with $\ZipParent = \OptNone$, $\ZipLeft = \ListEmpty$, and $\ZipRight =
  \ListEmpty$ respectively. For example:
  \[
     \MkZip{\ZipCurrent}{}{}{} \eqdef
     \Set{\ZipParent = \OptNone; \ZipLeft = \ListEmpty; \ZipCurrent; \ZipRight = \ListEmpty}
  \]
  \noindent
  Likewise, to destruct a zipper we write,
  \(
    \GetZip{\ZipCurrent}{\ZipperMeta}{\ZipLeft}{\ZipRight}
  \)
  where any part can be left out to ignore that portion of the zipper, but any
  included part must exist. For example,
  $\ZipperMeta = \GetZip{\ignoreArg}{\ZipperMeta'}{}{} \defiff
  \ZipGetParent = \OptSomeF{\ZipperMeta'}$.
\end{definition}

\begin{figure}
  \(
    \begin{tabular}{| c | c | l |}
      \hline
      Command: \FCommandMeta{} & Conditions: \PredListMeta & Def. of $\FCDenFG{\FCommandMeta}{\ContextExp}{\GContextExp} ~ \text{ when } \PredListMeta$
      \\ \Xhline{5\arrayrulewidth}
      \FComDown
      & $
        \begin{aligned}
          \ZipperMeta &= \GetZip{\FZipNodeSNS{\FSpecPathD}}{}{}{}
          \\
          \pair{\StrMeta}{\LogMeta} &= \FEDenDLE{\ExpMeta}
          \\
          \FSDirD &= \MapGet{\FSMeta}{\PathMeta}
        \end{aligned}
        $
      & $
        \begin{array}{lcl}
          \LogMeta' & = & \ListAppendLF{\LogMeta}{(\TxLReadF{(\FSDirD)}{\PathMeta})}
          \\
          \BMFullReturnF
            {(\EnvMeta,\PathCons{\PathMeta}{\StrMeta},\MkZip{\FZipNodeSND}{\ZipperMeta}{}{})}
            {\pair{\PathSetMeta \union (\PathCons{\PathMeta}{\StrMeta})}{\FSMeta}}
            {\LogMeta'}
            \span \span
        \end{array}
        \mkern-16mu
        $
      \\ \hline
      \FComUp
      & $
        \begin{aligned}
          \ZipperMeta &= \GetZip{\ignoreArg}{\ZipperMeta'}{}{}
          \\
          \ZipperMeta' &= \GetZip{\FZipNodeSNI{\FSpecPathD}}{}{}{}
        \end{aligned}
        $
      & $
        \begin{array}{lcl}
            \BMFullReturnF{(\EnvMeta,\FFPopD,\ZipperMeta')}{\GContextExp}{\LogEmpty}
            \span \span
        \end{array}
        $
      \\ \Xhline{3\arrayrulewidth}
      \FComIntoOpt
      & $
        \begin{aligned}
          \ZipperMeta &= \GetZip{\FZipNodeSNS{\FSpecOptD}}{}{}{}
        \end{aligned}
        $
      & $
        \begin{array}{lcl}
          \BMFullReturnZ{\MkZip{\FZipNodeSND}{\ZipperMeta}{}{}} \span \span
        \end{array}
        $
      \\ \hline
      \FComIntoPair
      & $
        \begin{aligned}
          \ZipperMeta &= \GetZip{\FZipNodeSNS{\FSpecPairD}}{}{}{}
        \end{aligned}
        $
      & $
        \begin{array}{lcl}
          \ContextMeta & = & (\FEnvLoc,\PathMeta,\MkZip{\FZipNodeSNS{\SpecMeta_1}}{}{}{})
          \\
          \ZipperMeta' & = & \MkZip{\FZipNodeSNS{\SpecMeta_1}}{\ZipperMeta}{}
            {\List{\FZipNodeF{\FEnvLoc\subst{\VarMeta}{\ContextMeta}}{\SpecMeta_2}}}
          \\
          \BMFullReturnZ{\ZipperMeta'} \span \span
        \end{array}
        \mkern-16mu
        $
      \\ \hline
      \FComIntoComp
      & $ \!
        \begin{aligned}
          \ZipperMeta &= \GetZip{\FZipNodeSNS{\FSpecCompD}}{}{}{}
          \\
          \pair{\ListAppendIF{h}{t}}{\LogMeta} &= \FEDenDLE{\ExpMeta}
        \end{aligned}
        $
      & $
        \begin{array}{lcl}
          r & = & \mapFunF{(\anonFunF{\StrMeta}{\FZipNodeF{\FEnvLoc\subst{\VarMeta}{\StrMeta}}{\SpecMeta}})}{t}
          \\
          \ZipperMeta' & = &
            \MkZip{\FZipNodeSNF{\FEnvLoc\subst{\VarMeta}{h}}{\SpecMeta}}
            {\ZipperMeta}{}
            {r}
          \\
          \BMFullReturnZL{\ZipperMeta'}{\LogMeta} \span \span
        \end{array}
        $
      \\ \hline
      \FComOut
      & $
        \begin{aligned}
          \ZipperMeta &= \GetZip{\ignoreArg}{\ZipperMeta'}{}{}
          \\
          \ZipperMeta' &\neq \GetZip{\FZipNodeSNI{\FSpecPathD}}{}{}{}
        \end{aligned}
        $
      & $
        \begin{array}{lcl}
          \BMFullReturnZ{\ZipperMeta'}\span \span
        \end{array}
        $
      \\ \Xhline{3\arrayrulewidth}
      \FComNext
      & $ \ZipperMeta = \GetZip{c'}{\ZipperMeta'}{l}{(\ListAppendIF{c}{r})} $
      & $
        \begin{array}{lcl}
          \ZipperMeta'' &=& \MkZip{c}{\ZipperMeta'}{(\ListAppendIF{c'}{l})}{r}
          \\
          \BMFullReturnZ{\ZipperMeta''}\span \span
        \end{array}
        $
      \\ \hline
      \FComPrev
      & $ \ZipperMeta = \GetZip{c'}{\ZipperMeta'}{(\ListAppendIF{c}{l})}{r} $
      & $
        \begin{array}{lcl}
          \ZipperMeta'' &=& \MkZip{c}{\ZipperMeta'}{l}{(\ListAppendIF{c'}{r})}
          \\
          \BMFullReturnZ{\ZipperMeta''} \span \span
        \end{array}
        $
      \\ \Xhline{3\arrayrulewidth}
      \FComStoreFileF{\! \ExpMeta}
      & $
        \begin{aligned}
          \ZipperMeta &= \GetZip{\FZipNodeSNI{\FSpecFile}}{}{}{}
          \\
          \pair{\StrMeta}{\LogMeta} &= \FEDenDE{\ExpMeta}
        \end{aligned}
        $
      & $
        \begin{array}{lcl}
          \pair{\FSMeta'}{\LogMeta'} & = & \FFMkFileD
          \\
          \BMFullReturnG{\pair{\PathSetMeta}{\FSMeta'}}{\ListAppendLF{\LogMeta}{\LogMeta'}} \span \span
        \end{array}
        $
      \\ \hline
      \FComStoreDirD
      & $
        \begin{aligned}
          \ZipperMeta &= \GetZip{\FZipNodeSNI{\FSpecDir}}{}{}{}
          \\
          \pair{\SetMeta}{\LogMeta} &= \FEDenDE{\ExpMeta}
        \end{aligned}
        $
      & $
        \begin{array}{lcl}
          \pair{\FSMeta'}{\LogMeta'} & = & \FFMkNodeD
          \\
          \BMFullReturnG{\pair{\PathSetMeta}{\FSMeta'}}{\ListAppendLF{\LogMeta}{\LogMeta'}} \span \span
        \end{array}
        $
      \\ \hline
      \FComCreatePath
      & $
        \begin{aligned}
          \ZipperMeta &= \GetZip{\FZipNodeSNS{\FSpecPathD}}{}{}{}
          \\
          \pair{\StrMeta}{\LogMeta} &= \FEDenDLE{\ExpMeta}
        \end{aligned}
        $
      & $
        \begin{array}{lcl}
          \pair{\FSMeta'}{\LogMeta'} & = & \FFAddNodeD
          \\
          \LogMeta'' & = &  \ListAppendLF{\LogMeta}{
            \ListAppendLF{(\TxLReadF{\MapGet{\FSMeta}{\PathMeta}}{\PathMeta})}
              {\LogMeta'}}
          \\
          \BMFullReturnG{\pair{\PathSetMeta}{\FSMeta'}}{\LogMeta''} \span \span
        \end{array}
        $
      \\ \hline
    \end{tabular}
  \)
  \jdl{Reviewer 3 suggests that a figure with the type signatures for auxiliary
  functions (like pop, make\_file etc) and perhaps a paragraph with short
  explanations are in order. We do have such a figure (though it doesn't include
  pop) in the appendix, but, while I think pop is pretty clear, I do sorta agree
  that we should address it a bit more in the enumeration. At the same time, the
  details are of these functions are not really interesting and their name and
  brief explanation should give readers all the intuition they need unless they
  want to look at proofs.}
  \caption{\FCommandMeta{} Command Semantics}
  \label{fig:forest-commands}
\end{figure}

\jdl{Fix flow and intro of (and just look at) this paragraph}
\pnote{Command Semantics: Invariants + Locality.}
The two key invariants that hold during execution of any command are
(1) that the file system remains well-formed (\defref{def:fs-wf})
and (2) that if
$\FCDenFG{\FCommandMeta}
  {(\ignoreArg,\PathCons{\PathMeta}{\StrMeta},\ignoreArg)}{\pair{\ignoreArg}{\FSMeta}}
  = \BMFullReturnF{(\ignoreArg,\PathCons{\PathMeta'}{\StrMeta'},\ignoreArg)}{\pair{\ignoreArg}{\FSMeta'}}{\ignoreArg}
$
and $\PathMeta \in \FSMeta$, then $\PathMeta' \in \FSMeta$.
The first property states that no command can make a well-formed file system
ill-formed. The second states that, as we traverse the zipper, we maintain a
connection to the real file system. It is important that only the parent of the
current file system node is required to exist, which enables constructing new
portions of the filestore. A central design choice that underpins the semantics
is that each command acts \emph{locally} on the current zipper and does not
require further context. This makes the cost of the operation apparent and, as
in Incremental Forest~\cite{dilorenzo2016incremental}, facilitates partial
loading and storing.

\jdl{We could expand on why this connection is so weird (i.e. why not just
require that the current path is always in the FS?), which is for the sake of
options.}

%
%
%
\pnote{Semantics Description: Forest Command Intro.}
\figref{fig:forest-commands} defines the semantics of Forest Commands. As
illustrated in the top row of the table, each row should be interpreted as
defining the meaning of evaluating the command in a given local and global
context, $\ContextExp$ and $\GContextExp$, provided the conditions hold. The
denotation function is partial, being undefined if none of the rows apply.
Intuitively, a command is undefined when it is used on a malformed file system
with respect to the specification, or when it is ill-typed---i.e. used on an
unexpected zipper state. Operationally, the semantics of each command can be
understood as follows:

\jdl{Reviewer 3 wonders what we mean by 'duals' even though we say precisely
what we mean after mentioning it even using a colon to indicate this
fact} 

\begin{itemize}
  \item \FComDown{} and \FComUp{} are duals: The first traverses the zipper into a path
  expression, simultaneously moving us down in the filestore, while the other
  does the reverse. Additionally, \FComDown{} queries the file system,
  producing a \TxLReadName{}.
  \item $\FComStyle{Into}$ and $\FComOut$ are duals: The first
  traverses the zipper into its respective type of specification, while the
  second moves back out to the parent node. Additionally, their subexpressions
  may produce logs.
  \\
  For dependent pairs, we update the environment of the second
  child with a context constructed from the first specification.
  \\
  For comprehensions, the traversal requires the expression evaluation to be
  non-empty, and constructs a list of children with the same specification,
  but environments with different mappings for $\VarMeta$, before moving to
  the first child.
  \item \FComNext{} and \FComPrev{} are duals: The first traverses the zipper
  to the right sibling and the second to the left sibling.
  \item \FComStoreFileD{}, \FComStoreDirD{}, and \FComCreatePath{} all update
  the file system, leaving the zipper untouched. The functions they call out to
  all close the file system to remain well-formed and their definitions can be
  found in \appref{fig:forest-helpers}. These functions
  produce logs recording their effects.
  \\
  For \FComStoreFileD{}, $\ExpMeta$ must evaluate to a string, $\StrMeta$,
  after which the command turns the current file system node into a file
  containing $\StrMeta$.
  \\
  For \FComStoreDirD{}, $\ExpMeta$ must evaluate to a string set, $\SetMeta$,
  after which the command turns the current file system node into a directory
  containing that set. If the node is already a directory containing
  $\SetMeta'$, then any children in $\SetMeta' \setminus \SetMeta$ are
  removed, any children in $\SetMeta \setminus \SetMeta'$ are added (as empty
  files) and any children in $\SetMeta \inter \SetMeta'$ are untouched.
  \\
  For \FComCreatePath{}, the current node is turned into a directory
  containing the path that the path expression points to. The operation is
  idempotent and does the minimal work required: If the current node
  is already a directory, then the path is added. If the path was already
  there, then \FComCreatePath{} is a no-op, otherwise it will map to an empty
  file. 
\end{itemize}

\begin{figure}[t]
  \(
    \begin{tabular}{| c | c | l |}
      \hline
      Expression: \FExpMeta{{}} & Conditions: \PredListMeta & Def. of $\FEDenDE{\FExpMeta}  ~ \text{ when } \PredListMeta$
      \\ \Xhline{5\arrayrulewidth}
      \FExpFetchFile
      & $
        \begin{aligned}
          \ZipperMeta &= \GetZip{\FZipNodeSNI{\FSpecFile}}{}{}{}
          \\
          \FSFileD &= \MapGet{\FSMeta}{\PathMeta}
        \end{aligned}
        $
      & $
        \begin{array}{lcl}
          \ReturnF{\StrMeta}{\List{\TxLReadF{(\FSFileD)}{\PathMeta}}}
          \span \span
        \end{array}
        $
      \\ \hline
      \FExpFetchDir
      & $
        \begin{aligned}
          \ZipperMeta &= \GetZip{\FZipNodeSNI{\FSpecDir}}{}{}{}
          \\
          \FSDirD &= \MapGet{\FSMeta}{\PathMeta}
        \end{aligned}
        $
      & $
        \begin{array}{lcl}
          \ReturnF{\SetMeta}{\List{\TxLReadF{(\FSDirD)}{\PathMeta}}}
          \span \span
        \end{array}
        $
      \\ \hline
      \FExpFetchPath
      & $
        \begin{aligned}
          \ZipperMeta &= \GetZip{\FZipNodeSNS{\FSpecPathD}}{}{}{}
        \end{aligned}
        $
      & $
        \begin{array}{lcl}
          \FEDenDLE{\ExpMeta}
          \span \span
        \end{array}
        $
      \\ \hline
      \FExpFetchComp
      & $
        \begin{aligned}
          \ZipperMeta &= \GetZip{\FZipNodeSNS{\FSpecCompD}}{}{}{}
        \end{aligned}
        $
      & $
        \begin{array}{lcl}
          \FEDenDLE{\ExpMeta}
          \span \span
        \end{array}
        $
      \\ \hline
      \FExpFetchOpt
      & $
        \begin{aligned}
          \ZipperMeta &= \GetZip{\FZipNodeSNI{\FSpecOptD}}{}{}{}
        \end{aligned}
        $
      & $
        \begin{array}{lcl}
          \ReturnF{\PathMeta \in \FSMeta}{\List{\TxLReadF{\MapGet{\FSMeta}{\PathMeta}}{\PathMeta}}}
          \span \span
        \end{array}
        $
      \\ \hline
      \FExpFetchPred
      & $
        \begin{aligned}
          \ZipperMeta &= \GetZip{\FZipNodeSNS{\FSpecPredD}}{}{}{}
        \end{aligned}
        $
      & $
        \begin{array}{lcl}
          \FEDenDLE{\ExpMeta}
          \span \span
        \end{array}
        $
      \\ \Xhline{3\arrayrulewidth}
      \FExpRunCD
      & $
        \begin{aligned}
          \pair{\ContextMeta}{\LogMeta} &= \FEDenDE{\ExpMeta}
          \\
          (\ContextMeta',\ignoreArg,\LogMeta') &= \FCDenFG{\FNavigationMeta}{\ContextMeta}{\GContextExp}
        \end{aligned}
        $
      & $
        \begin{array}{lcl}
          \ReturnF{\ContextMeta'}{\ListAppendLF{\LogMeta}{\LogMeta'}}
          \span \span
        \end{array}
        $
      \\ \hline
      \FExpRunED
      & $
        \begin{aligned}
          \pair{\ContextMeta}{\LogMeta} &= \FEDenDE{\ExpMeta}
          \\
          \pair{\ValMeta}{\LogMeta'} &= \FEDenF{\FExpMeta}{\ContextMeta}{\GContextExp}
        \end{aligned}
        $
      & $
        \begin{array}{lcl}
          \ReturnF{\ValMeta}{\ListAppendLF{\LogMeta}{\LogMeta'}}
          \span \span
        \end{array}
        $
      \\ \hline
      \FFunVerify
      & $\btrue$
      & $
        \begin{array}{lcl}
        \pair{\PathMeta'}{\ZipperMeta'} &=& \FFGoToRootD
          \\
          \FPConsistentF{(\PathMeta',\ZipperMeta')}{\GContextExp}
          \span \span
        \end{array}
        $
      \\ \hline
    \end{tabular}
  \)
  
  \caption{Expression Semantics}
  \label{fig:forest-expressions}
\end{figure}

%
%
%
\pnote{Forest Command Wrap-up + Forest Expression semantics intro.}
With that, we have covered the semantics of all of the Forest Commands, but 
their subexpressions remain.
%
%
The semantics of non-standard expressions is given in
\figref{fig:forest-expressions}. The interpretation of each row is the same as
for commands. There is one $\FComStyle{Fetch}$ expression per specification
except for pairs, which have no useful information available locally. Since a
pair is defined in terms of its sub-specifications, we must navigate to them
before fetching information from them. This design avoids incurring the
cost of eagerly loading a large filestore.

\pnote{Fetch Semantics.}
Fetching a file returns the string contained by the file at the current path.
For a directory, we get the list of its children. Both of these log
\TxLReadName{}s since they inspect the file system. For a path specification,
the only locally available information is the actual path. For a
comprehension, we return the string set. For an option, we find out whether
the current path is in the file system or not and log a \TxLReadName{}
regardless. Finally, for a predicate, we determine if it holds.

\jdl{Reviewer 3 is confused about the fact that fetching a file returns a string
rather than a file-handle (and assumes that in the implementation, it isn't so).
I'm pretty convinced we shouldn't do anything about this.}

\pnote{Run Expression Semantics.}
There are two $\FComStyle{Run}$ expressions. The subexpression,
$\ExpMeta$, must evaluate to a local context. These can only come from a
dependent pair, which means that $\FComStyle{Run}$s can only occur as
subexpressions of specifications. We utilize them by performing traversals
($\FExpRunCD$) and evaluating Forest expressions ($\FExpRunED$) in the input
context. For example, a filestore defined by a file \texttt{index.txt} and a
set of files listed in that index could be described as follows:
\begin{align*}
  &\FSpecPairF
    {index}
    {\FSpecPathF{"index.txt"}{\FSpecFile}}
    {\\&
    \FSpecCompF
      {\FSpecPathF{\VarMeta}{\FSpecFile}}
      {\VarMeta}
      {\texttt{lines\_of}~
        (\FExpRunF
          {\FExpFetchFile}
          {(\FExpRunF{\FComDown}{index})}
        )}}
\end{align*}
where \texttt{lines\_of} maps a string to a string set by splitting it by lines.

\pnote{Verification and Partial Consistency Semantics.}
Finally, \FFunVerify{} checks the partial consistency of the traversed part of
the filestore---i.e. whether it conforms to our specification. Unfortunately,
checking the entire filestore, even incrementally can be very expensive and,
often, we have only performed some local changes and thus do not need the full
check. Partial consistency is a compromise wherein we only check the portions of
the filestore that we have traversed, as denoted by the path set. This ensures
that the cost of the check is proportional to the cost of the operations we have
already run. Partial consistency is formally defined in the next subsection,
which among other properties, details the connection between partial and full
consistency.

\jdl{Reviewer 3 suggests that examples may be in order too (though they only
requested the type signatures)}

\subsection{Properties}\label{subsec:forest-properties}
%
%
%
\pnote{Properties intro.}
This section establishes properties of the \TXF{} core
calculus: consistency and partial consistency, equational identities
on commands, and round-tripping laws.

\begin{figure}
  \(
    \begin{tabular}{| c | c | l |}
      \hline
      Spec: \SpecMeta & 
      Conditions: \PredListMeta & 
      Def. of 
        $\FPConsistentZ{\GetZip{\pair{\EnvMeta}{\SpecMeta}}{}{}{} \text{ as } \ZipperMeta}  
        ~ \text{ when } \PredListMeta$
      \\ \Xhline{5\arrayrulewidth}
      $\ignoreArg$
      & $\PathMeta \notin \PathSetMeta $
      & $
        \begin{array}{lcl}
          \FPConstRetE{\btrue}{\bfalse}
          \span \span
        \end{array}
        $
      \\ \Xhline{3\arrayrulewidth}
      $\FSpecFile{}$
      & $\PathMeta \in \PathSetMeta $
      & $
        \begin{array}{lcl}
          \FPConstRetR{\MapGet{\FSMeta}{\PathMeta} = \FSFileF{\ignoreArg}}{\btrue}
          \span \span
        \end{array}
        $
      \\ \hline
      $\FSpecDir{}$
      & $\PathMeta \in \PathSetMeta $
      & $
        \begin{array}{lcl}
          \FPConstRetR{\MapGet{\FSMeta}{\PathMeta} = \FSDirF{\ignoreArg}}{\btrue}
          \span \span
        \end{array}
        $
      \\ \hline
      $\FSpecPathD$
      & $\PathMeta \in \PathSetMeta $
      & $
        \begin{array}{lcl}  
        \pair{\StrMeta}{\LogMeta} & = & \FEDenDE{\ExpMeta}
        \\
          \FPConstRetF{\MapGet{\FSMeta}{\PathMeta} = \FSDirF{\ignoreArg}}{\btrue}
          {\ListAppendLF{\LogMeta}{(\TxLReadFS)}}
          ~ \bland \span \span
        \\
          \FPConsistentF{\pair{\PathConsS}
            {\MkZip{\pair{\EnvMeta}{\SpecMeta}}{\ZipperMeta}{}{}}}
          {\GContextExp}
          \span \span
        \end{array}
        $
      \\ \hline
      $\FSpecPairD$
      & $\PathMeta \in \PathSetMeta $
      & $
        \begin{array}{lcl}
          \ContextMeta & = & (\EnvMeta,\PathMeta,\MkZip{\pair{\EnvMeta}{\SpecMeta_1}}{}{}{})
          \\
          \EnvMeta' & = & \EnvMeta\subst{\VarMeta}{\ContextMeta}
          \\
          \FPConsistentZ{\MkZip{\pair{\EnvMeta}{\SpecMeta_1}}{\ZipperMeta}{}
          {\List{\pair{\EnvMeta'}{\SpecMeta_2}}}}
          ~ \bland \span \span
        \\
          \FPConsistentZ{\MkZip{\pair{\EnvMeta'}{\SpecMeta_2}}{\ZipperMeta}
          {\List{\pair{\EnvMeta}{\SpecMeta_1}}}{}}
          \span \span
        \end{array}
        $
      \\ \hline
      $\FSpecCompD$
      & $\PathMeta \in \PathSetMeta $
      & $
        \begin{array}{lcl}
          \pair{\SetMeta}{\LogMeta'} & = & \FEDenDE{\ExpMeta}
          \\
          \FPConstRetF{\BoolMeta_1}{\BoolMeta_2}{\LogMeta} & = & 
            \bigwedge\limits_{\ValMeta \in \SetMeta}
            \FPConsistentZ{\MkZip{\pair{\EnvMeta\subst{\VarMeta}{\ValMeta}}{\SpecMeta}}{\ZipperMeta}{}{}}
          \\ 
          \FPConstRetF{\BoolMeta_1}{\BoolMeta_2}{\ListAppendLF{\LogMeta'}{\LogMeta}}
            \span \span
        \end{array}
        $
      \\ \hline
      $\FSpecOptD$
      & $\PathMeta \in \PathSetMeta $
      & $
        \begin{array}{lcl}
          \FPConstRetR{\PathMeta \notin \FSMeta}{\btrue} ~ \blor ~ 
          \span \span
          \\ 
          \FPConsistentZ{\MkZip{\pair{\EnvMeta}{\SpecMeta}}{\ZipperMeta}{}{}}
          \span \span
        \end{array}
        $
      \\ \hline
      $\FSpecPredD$
      & $\PathMeta \in \PathSetMeta $
      & $
        \begin{array}{lcl}
          \pair{\BoolMeta}{\LogMeta} & = & \FEDenDE{\ExpMeta}
          \\
          \FPConstRetD{\BoolMeta} 
          \span \span          
        \end{array}
        $
      \\ \hline
    \end{tabular}
  \)
  \[
    \begin{array}{lcl}
      \pair{\pair{\bfalse}{\ignoreArg}}{\LogMeta} ~ \bland ~ \ignoreArg & \eqdef &  \pair{\pair{\bfalse}{\bfalse}}{\LogMeta}
      \\
      \pair{\pair{\BoolMeta_1}{\BoolMeta_2}}{\LogMeta} ~ \bland ~ \pair{\pair{\BoolMeta_1'}{\BoolMeta_2'}}{\LogMeta'}
        & \eqdef &  \pair{\pair{\bandf{\BoolMeta_1}{\BoolMeta_1'}}{\bandf{\BoolMeta_2}{\BoolMeta_2'}}}
                          {\ListAppendLF{\LogMeta}{\LogMeta'}}
      \\
      \pair{\pair{\btrue}{\btrue}}{\LogMeta} ~ \blor ~ \ignoreArg & \eqdef &  \pair{\pair{\btrue}{\btrue}}{\LogMeta}
      \\
      \pair{\pair{\BoolMeta_1}{\BoolMeta_2}}{\LogMeta} ~ \blor ~ \pair{\pair{\BoolMeta_1'}{\BoolMeta_2'}}{\LogMeta'}
        & \eqdef &  \pair{\pair{\borf{\BoolMeta_1}{\BoolMeta_1'}}{\borf{\BoolMeta_2}{\BoolMeta_2'}}}
                          {\ListAppendLF{\LogMeta}{\LogMeta'}}

      \\
      \\
      \FPathCoverD & \defiff & \pairSnd ~ (\pairFst ~ (\FPConsistentD))
    \end{array}
  \]
  
  %
  \caption{Partial Consistency and Cover}
  \label{fig:partial-consistency}
\end{figure}

%
%
%
\pnote{Consistency.}
The formal definition of partial consistency is given in
\figref{fig:partial-consistency}. Intuitively, full consistency (\FConsistentName{})
captures whether a filestore conforms to its specification. For example, the
file system, $\FSMeta$, at $\PathMeta$ conforms to $\FSpecFile$ if and only if
$\MapGet{\FSMeta}{\PathMeta} = \FSFileF{\ignoreArg}$ and to $\FSpecPathD$ if
$\ExpMeta$ evaluates to $\StrMeta$ and $\FSMeta$ at
$\PathCons{\PathMeta}{\StrMeta}$ conforms to $\SpecMeta$. Partial consistency
(\FPConsistentName{}) then checks partial conformance (i.e. does the filestore
conform to part of its specification). \FPConsistentName{} returns two booleans
(and a log), the first describing whether the input filestore is consistent with
the input specification and the second detailing whether that consistency is
total or partial. The definition of full consistency is very similar to partial,
except that there are no conditions and the pathset is ignored. The properties
below describe the relationship between partial consistency and full
consistency. Their proofs can be found in
\appref{app:properties}.

\jdl{Perhaps describe some rule more closely}

%
%
%
\pnote{Consistency theorems.}
\begin{theorem}\torestate{
  \label{thm:const-to-pconst}
  Consistency implies partial consistency: \\
  $\forall \PathSetMeta. ~ \FConsistentD \implies \pairFst ~ (\FPConsistentD)$
}
\end{theorem}
\begin{theorem}\torestate{
  \label{thm:pconst-monotonicity}
  Partial Consistency is monotonic w.r.t. the path set:\\
  $\forall \PathSetMeta_1,\PathSetMeta_2. ~ \PathSetMeta_2 \subseteq \PathSetMeta_1
  \implies$\\
  $\pairFst ~ (\FPConsistentF{\pair{\EnvMeta}{\ZipperMeta}}{\pair{\PathSetMeta_1}{\FSMeta}}) \implies
  \pairFst ~ (\FPConsistentF{\pair{\EnvMeta}{\ZipperMeta}}{\pair{\PathSetMeta_2}{\FSMeta}})$\\
  $\band ~
  \pairSnd ~ (\FPConsistentF{\pair{\EnvMeta}{\ZipperMeta}}{\pair{\PathSetMeta_2}{\FSMeta}}) \implies
  \pairSnd ~ (\FPConsistentF{\pair{\EnvMeta}{\ZipperMeta}}{\pair{\PathSetMeta_1}{\FSMeta}})
  $
}
\end{theorem}
This theorem says that if $\PathSetMeta_1$ is partially consistent,
then any path set, $\PathSetMeta_2$, that is a subset of $\PathSetMeta_1$ will also be
partially consistent. Conversely, if the consistency of $\PathSetMeta_2$ is
total, $\PathSetMeta_1$ will also be totally consistent.

\begin{theorem}\torestate{
  \label{thm:pconst-eq-const}
  Given a specification $\SpecMeta$ and a path set $\PathSetMeta$ that
  covers the entirety of $\SpecMeta$, partial consistency is exactly full
  consistency: \\
  $\forall \PathSetMeta. ~ \exists \PathSetMeta'. ~
  \FPathCoverF{\pair{\PathMeta}{\ZipperMeta}}{\pair{\PathSetMeta'}{\FSMeta}} 
  ~ \band ~ \PathSetMeta' \subseteq \PathSetMeta
  \implies$\\
  \hspace*{2.1em}$\FConsistentD \iff \pairFst ~ (\FPConsistentD)$
}
\end{theorem}

\jdl{Make sure this new cover definition works well. It seems much simpler and cleaner.}

This theorem says that if the path set, $\PathSetMeta$ is a superset of one that
covers the entirety of the filestore, $\PathSetMeta'$, as defined in
\figref{fig:partial-consistency}, then the filestore is totally consistent
exactly when it is partially consistent. Intuitively, if a path set covers a
filestore then we can never encounter a path outside of the path set while
traversing the zipper.

\jdl{TODO: Currently, the theorems ignore the log information returned from
consistent and pconsistent, which is more convenient than having unnecessary
additions, but is perhaps not well motivated. Auxiliary functions may be a good
move.}

%
%
%
\pnote{Round-tripping laws and identities.}
Other properties of the language include identities of the form
$\areEquiv{\FCDen{\ComSeqF{\FComDown}{\FComUp}}}{\FCDen{\ComSkip}}$ where
$\undefeq$ denotes equivalence modulo log when defined. That is, either
$\FCDen{\ComSeqF{\FComDown}{\FComUp}}$ is undefined, or it has the same action
as $\FCDen{\ComSkip}$, barring logging. Additionally, we have proven
round-tripping laws in the style of lenses~\cite{focal-toplas}
stating, for example, that storing just loaded data is equivalent to \ComSkip{}.
Further identities and formal statements of the round-tripping laws can be
found in \appref{app:properties}.

%
%
%
\pnote{Segue into concurrency.}
This concludes the description of the syntax, semantics, and properties of
\TXF{}. So far, we have focused on a single thread of execution, but
to fulfill our goal of proving that multiple \TXF{} transactions are
serializable with respect to one another, we need to be able to model them
running concurrently.

\section{Concurrency Control}\label{sec:concurrency}
%
%
%
\pnote{Section Overview: Global semantics, denotational is serial, operational is serializable.} 
This section introduces the global semantics of \TXFFull{}, using both a
denotational semantics to concisely capture a serial semantics, and an
operational semantics to capture thread interleavings and concurrency. Then, we
state our serializability theorem by relating the two semantics.

\begin{figure}[t]
  \(
  \begin{array}{@{}r@{\;}c@{\;}l@{~~}l@{}}
  &&& \vdots \\
  \TimestampMeta & \in & \TimestampSet & \text{Timestamps} \\
  \GLogMeta & \in & \GLogSet & \text{Timestamped Logs} \\
  \ThreadMeta & \in & \ThreadSet
  & \eqdef
  \product{\EnvSet}{\product{\FSSet}{\product{\PathSet}{\product{\PathSetSet}{\product{\ZipperSet}{\CommandSet}}}}}
  \\
  \TxStateMeta & \in & \TxStateSet
  & \eqdef \product{\CommandSet}{\product{\TimestampSet}{\LogSet}} \\
  \TransactionMeta & \in & \TransactionSet
  & \eqdef \product{\ThreadSet}{\TxStateSet} \\
  \TxBagMeta & \in & \TxBagSet & \eqdef \mathit{\TransactionSet ~ Bag}
  \end{array}
  \)
  \caption{Global Semantics Additional Syntax}
  \label{fig:forest-tx-syntax}
\end{figure}

%
%
%
\pnote{Brief syntax addition.}
\figref{fig:forest-tx-syntax} lists the additional syntax used in this section.
Timestamped logs are the logs of the global semantics. They are identical
to local logs except that each entry also contains a timestamp signifying when
it was written to the log.

\jdl{Reviewer 3 quabbles about 'local context' since it it overloaded with local
context in the previous section. They are right. After changing previous
section, check to make sure this makes sense.}

Each \ThreadSet{} is captured by its local context, which, along with its
transactional state, \TxStateSet{}, denotes a \TransactionSet{}. The
transactional state has 3 parts: (1) The full command that the transaction is
executing; (2) The time when the transaction started; and (3) The
transaction-local log recorded so far.

\jdl{TODO6: Here or earlier, add something about what \ignoreArg{} means}
%
%
%
\pnote{Denotational Semantics description + has no concurrency due to being serial.}
Our global denotational semantics is defined as follows:
\begin{align*}
  \FGSDenT{\TxMkF{\ThreadFS{\ignoreArg}}{\ignoreArg}} ~ &\eqdef ~
    \texttt{project\_fs} ~ (\FCDenDGE{\CommandMeta})\\
  \FGDenD ~ &\eqdef ~ \foldFunF{\FSMeta}{\ListMeta}{\FGSDenName}
\end{align*}
The denotation of one or more transactions is a function on file systems. For a
single transaction, it is the denotation of the command with the encapsulated
context except for the file system which is replaced by the input. For a list of
transactions, it is the result of applying the local denotation function in some
serial order. Note that the denotation of a transaction is precisely the
denotation of a program, $\FPDenName{}$, which can be lifted to multiple
programs by folding. The key point to note about this semantics is that there is
no interleaving of transactions. By definition, the transactions are run
sequentially. While this ensures serializability, it also does not allow for any
concurrency.

%
%
%
\pnote{Local Operational Semantics description.}
We will instead use an operational semantics that more easily models thread
interleaving and prove that it is equivalent to the denotational semantics.
First, we introduce an operational semantics for local commands. This semantics
is standard for IMP commands, but for Forest Commands, it uses the denotational
semantics, considering each a single atomic step, as seen below:

\jdl{As per Reviewer 3's suggestion, it's possible that we should qualify
'equivalent' here, since the sense in which it is equivalent is perhaps unusual.}

\infrule
{ ((\EnvMeta',\PathMeta',\ZipperMeta'),\pair{\PathSetMeta'}{\FSMeta'},\LogMeta)
   = \FCDenDGE{\FCommandMeta}}
{ \FAToACD{}{'}{\FCommandMeta}{\ComSkip}}

%
%
%
\pnote{Global Operational Semantics description.}
Next, we can construct the global operational semantics, as seen in
\figref{fig:forest-tx-semantics}. The global stepping relation is between two
global contexts which have three parts: A global file system, a global
log, and a thread pool, or bag of transactions.

\begin{figure}
  \infrule
  {\FStepTo{\FContMan{\ThreadMeta}}{\LogMeta'}{\FContMan{\ThreadMeta'}}}
  {\TxTToT
    {\TxThMDD}
    {\TxMkF{\ThreadMeta'}{\TxSL{\ListAppendLF{\LogMeta}{\LogMeta'}}}}
  }
  \vspace{1em}
  \infrule
  {\texttt{is\_Done?}~\ThreadMeta
   \quad \TxFCheckLogD
   \quad \GFSMeta' = \TxFMergeD
   \quad \GLogMeta' = \ListAppendLF{\GLogMeta}{(\TxAddTSF{\TxFreshTS}{\LogMeta})}
  }
  {\TxStepTo{\TxContUT{\TxThMD{\TxSD}}}
  {\TxCont{\GFSMeta'}{\GLogMeta'}{\TxBagMeta}}}

  \vspace{1em}

  \infrule
  { \texttt{is\_Done?}~\ThreadMeta
    \quad \neg (\TxFCheckLogD)
    \quad \TimestampMeta' = \TxFreshTS
    \quad \pair{\ZipperMeta'}{\PathMeta'} = \FFGoToRootD
  }
  {\TxTToT{\TxThMD{\TxSD}}{\TxStartF{\GFSMeta}{\TimestampMeta'}{'}}}

  \caption{Global Operational Semantics}
  \label{fig:forest-tx-semantics}
\end{figure}

There are only three actions that the global semantics can take:
\begin{enumerate}
  \item A transaction can step in the local
        semantics and append the resulting log.
  \item A transaction that is done, and does not conflict with previously
        committed transactions, can commit. It must check that none of its
        operations conflicted with those committed since its start. Conflicts
        occur when the transaction read stale data. Then, it will update the
        global file system according to any writes performed. Finally, the
        transaction will leave the thread pool. The definitions of
        \TxFCheckLogName{} and \TxFMergeName{} can be found in
        \figref{fig:merge-check-def}.
  \item A transaction that is done, but conflicts with previously committed
        transactions, cannot commit and instead has to restart. It does this by
        getting a fresh timestamp and resetting its log and local context.
\end{enumerate}

\begin{figure}
    \(
    \begin{array}{@{}r@{\;}c@{\;}l}
      \TxFMergeD & \eqdef & \foldFunF{\GFSMeta}{\LogMeta}{\TxFUpdateName}
      \\ \\
      \TxFUpdateF{\FSMeta}{(\TxLReadD)} & \eqdef & \FSMeta \\
      \TxFUpdateF{\FSMeta}{(\TxLWriteFileF{\ignoreArg}{\ContentsMeta}{\PathMeta})} & \eqdef
        & \TxFCloseF{(\FSMeta\subst{\PathMeta}{\ContentsMeta})} \\
      \TxFUpdateF{\FSMeta}{(\TxLWriteDirF{\ignoreArg}{\ContentsMeta}{\PathMeta})} & \eqdef
        & \TxFCloseF{(\FSMeta\subst{\PathMeta}{\ContentsMeta})}
      \\ \\
      \TxFCheckLogD & \eqdef
        & \forall \PathMeta' \in \TxFExtractPathsF{\LogMeta}. ~
        \forall (\TimestampMeta',\LogEntryMeta) \in \GLogMeta. \\
        && \TimestampMeta' < \TimestampMeta
        ~ \bor ~ \neg (\TxFConflictPathF{\PathMeta'}{\LogEntryMeta})
      \\
      \TxFConflictPathF{\PathMeta'}{(\TxLReadII)} & \eqdef & \bfalse \\
      \TxFConflictPathF{\PathMeta'}{(\TxLWriteFileII)} & \eqdef &  \TxFSubPathR \\
      \TxFConflictPathF{\PathMeta'}{(\TxLWriteDirII)} & \eqdef &  \TxFSubPathR
      \\ \\
      \TxFExtractPathsF{\ListEmpty} & \eqdef & \SetEmpty \\
      \TxFExtractPathsF{(\ListAppendIF{\TxLReadII}{tl})} & \eqdef
      & \Set{\PathMeta} \union (\TxFExtractPathsF{tl}) \\
      \TxFExtractPathsF{(\ListAppendIF{\TxLWriteFileII}{tl})} & \eqdef
      & \Set{\PathMeta} \union (\TxFExtractPathsF{tl}) \\
      \TxFExtractPathsF{(\ListAppendIF{\TxLWriteDirII}{tl})} & \eqdef
      & \Set{\PathMeta} \union (\TxFExtractPathsF{tl})
     \end{array}
    \)
  \jdl{Perhaps swap this to code-mode, and try to fit it in two mini-pages to
  save some space?}
  \caption{\TxFMergeName{} and \TxFCheckLogName{}}
  \label{fig:merge-check-def}
\end{figure}

\pnote{Global Operational Semantics has arbitary interleaving, but can roll-back
+ Completeness.}
In the operational semantics, thread steps can be interleaved arbitrarily, but
changes will get rolled back in case of a conflict. Furthermore, while Forest
Commands are modeled as atomic for simplicity, finer granularity would not
affect our results.

%
%
%
\pnote{Operational Semantics is serializable.}
With a useful global semantics, where transactions are run
concurrently, we now aim to prove that our semantics guarantees
serializability. The theorem below captures this property by
connecting the operational and denotational semantics:

\begin{theorem}[Serializability]\torestate{
  \label{thm:serializability}
  Let $\GFSMeta, \GFSMeta'$ be file systems, $\GLogMeta,\GLogMeta'$ be global
  logs, and $\TxBagMeta$ a thread pool such
  that $\forall \TransactionMeta \in \TxBagMeta. ~ \TxFInitialDG$, then:
  \\ \\
  $\TxMStepTo{\TxContAll{}}{\TxCont{\GFSMeta'}{\GLogMeta'}{\SetEmpty}}
    \implies
    \exists \ListMeta \in Perm(\TxBagMeta). ~ \FGDenF{\ListMeta}{\GFSMeta} = \GFSMeta'
  $
  \\ \\
  where $\TxMStep$ is the reflexive, transitive closure of $\TxStep$.
  }
\end{theorem}
The serializability theorem states that given a starting file system and a
thread pool of starting threads, if the global operational semantics commits
them all, then there is some ordering of these threads for which the global
denotational semantics will produce the same resulting file system. Note that
although it is not required by the theorem, the commit order is one such
ordering. Additionally, though not explicitly stated, it is easy to see that any
serial schedule that is in the domain of the denotation function is realizable
by the operational semantics. See \appref{app:serializability} for the proof.

\pnote{Segue to implementation.}
The prototype system described in the next section implements the local
semantics from the previous section along with this global semantics, reducing
the burden of writing correct concurrent applications. 

\section{Implementation}\label{sec:implementation}

\pnote{Section overview: Prototype, CMS example, and Surface Syntax}
This section describes our prototype implementation of \TXFFull{} as an embedded
domain-specific language and library in OCaml. We expand on the Course
Management example from \secref{sec:overview} and briefly touch on our
simplified surface syntax.

\jdl{We don't really say much about our prototype implementation... We could
write something about lines of code, though it's not so interesting... or...?}

%
%
%
\pnote{CMS example has additional features including a grading queue + Gradescope.}
We have implemented a simple course management system similar to the running
example from \secref{sec:overview}. It has several additional facilities beyond
renormalization, including computing various statistics about students or
homeworks and changing rubrics while automatically updating student grades
accordingly. The most interesting piece of the example is based on our
experience with a professional grading system which uses a queue from which
graders can get new problems to grade. Unfortunately, this system did not
adequately employ concurrency control, resulting in duplicated work. Using
\TXF{}, we implemented a simple grading queue where graders can add and retrieve
problems, which, with effectively no effort on our part, does not suffer from
such concurrency issues.

%
%
\pnote{Surface Syntax + Global semantics difference}
The embedded language in our prototype implementation implements almost
precisely the language seen in \secref{sec:forest}. Additionally, we provide a
surface syntax (as seen in \secref{sec:overview} and papers on the earlier
versions of Forest~\cite{forest-icfp:fisher+,dilorenzo2016incremental}) for
specifications that compiles down to the core calculus seen in
\secref{sec:forest}. This specification can then be turned into a zipper by
initiating a transaction. The majority of the commands and expressions seen in
the core semantics are then exposed as functions in a library. Additionally,
there is a more ad hoc surface command language that resembles the surface
syntax and parallels the behavior of the core language. Finally, the global
semantics looks slightly different from the one in \secref{sec:concurrency},
though this should not affect users in any way and the minor variant has also
been proven correct. We provide a simple shell for interacting with filestores,
which makes it significantly easier to force conflicts and test the concurrent
semantics.

\jdl{Reimplement the transactional semantics such that the caveat is unnecessary.}

\section{Related Work}\label{sec:related}
%
%
%
\pnote{Section Overview: General PADS project + Zippers + transaction semantics
+ transactional file systems.}
\TXFFull{} builds on a long line of work in ad hoc data processing. Its
semantics is designed around zippers as filestore representations, which is seen
in previous work on Zipper-based file systems~\cite{OlegZFSTalk}. There is prior
work on the semantics of transactions. Significant work has been done on
transactional file systems, which are file systems with concurrency guarantees.

\paragraph*{Ad hoc data processing.}
\pnote{PADS handles single ad hoc files and provides living documentation,
error-checking, parsers, deparsers, etc.}
PADS~\cite{fisher+:pads} (Processing Ad hoc Data Streams) is the first
declarative domain-specific language designed to deal with ad hoc data. It
allows users to write declarative specifications describing the structure of a
file and uses them to generate types, transformations between on-disk and
in-memory representations with robust error handling, along with various
statistical analysis tools.

\pnote{Forest extends this to full filestores and provides lens laws.}
Forest~\cite{forest-icfp:fisher+} extends the concept of PADS to full filestores
and additionally provides formal guarantees about the generated transformations
in the form of bidirectional lens laws. Forest was implemented in Haskell and
relied on its host language's laziness to only load required data.
Unfortunately, it was not necessarily obvious to users when they might
inadvertently load the whole filestore. For example, checking if there were
errors at any level would load everything below that level.

\pnote{Incremental Forest extends this with delays which allow incremental processing.}
Incremental Forest~\cite{dilorenzo2016incremental} attempted to mitigate this
issue by introducing delays to make explicit the precise amount of loading
performed by any action. It introduced a cost semantics to precisely
characterize the cost of any such action for varied, user-defined notions of
cost. However, concurrency remains challenging.


%
\pnote{How does this relate to TxForest?}
Unfortunately, we observe that in many fields where ad hoc data processing is
common, there is a pervasive need for concurrency control both for single user
parallelization and multi-user filestores. \TXFFull{} is designed around this
idea, providing serializable transactions by default.
Further, the zipper abstraction of \TXF{} is designed to provide incrementality
automatically. With \TXF{}, users can write standard Forest specifications
without considering size or delays, and the locality-focused zipper and
semantics design will enforce incrementality by only loading the minimal amount
of data necessary.

\paragraph*{Zippers.}
%
%
%
\pnote{Huet's Zippers and Oleg's ZFS.}
Zippers were first introduced in the literature by
Huet~\cite{Huet:1997:ZIP:969867.969872} as an elegant data structure for
traversing and updating a functional tree. There has been much work
studying zippers since, though the closest to our use case is Kiselyov's Zipper
file system~\cite{OlegZFSTalk}. Kiselyov builds a small functional file system
with a zipper as its core abstraction. This file system offers a simple
transaction mechanism by providing each thread their own view of the file
system, but lacks formal guarantees and a formal framework with which to prove
them. In contrast, \TXFFull{} uses type-based specifications describing the structure
and invariants of a filestore. Further, we present a formal syntax and semantics
for our core language, a model of concurrency, and a proof of serializability.

\paragraph*{Transaction Semantics.}
\pnote{AtomsFamily of languages with a cool type-and-effect system.}
Moore and Grossman~\cite{Moore:2008:HSO:1328438.1328448} present a family of
languages with software transactions and different semantics, investigating
parallelism in these languages and necessary restrictions to ensure correctness
in the presence of weak isolation. Additionally, they provide a type-and-effect
system which ensures the serializability of well-typed programs. At a high level, they
describe what the core of a language used to write concurrent programs might
look and act like, including constructs like spawning threads or atomic
sections. In contrast, our transactional semantics is simpler and specific to
our domain, describing a simpler transaction manager for ensuring
serializability among \TXF{} threads.

\jdl{TODO: Add a few more papers in this vein, for example, SPJ's Composable
Memory Transactions and perhaps Abadi's work (which is similar), and check if
something interesting cites this paper.}

\paragraph*{Transactional File Systems.}
%
%
%
\pnote{Transactional file systems are a large line of work that has some clear
relation, but is not worth mentioning in detail?}
There has been significant work on transactional file
systems~\cite{Escriva:2016:DIW:2930611.2930642,Garcia:1998:PFT:319195.319224,Liskov:2004:TFS:1133572.1133592,Schmuck:1991:ETQ:121132.121171}.
All of this work starts at a lower level than \TXFFull{}, providing
transaction support for file system commands. We, instead, provide
transactions from the perspective of the higher level application,
easily allowing an arbitrary high-level computation to be aborted or
restarted if there is a conflict at the file system level.

\jdl{I'm not quite sure what else to say about transactional file systems if anything}

\jdl{Any other important related work that I'm forgetting? Do we want to include
LINQ or Type Providers or Ntzik + Philippa? I feel like LINQ and Type Providers
don't apply quite as much anymore since we don't generate types in the same sort
of way, but perhaps that's too shallow.}


\section{Conclusion}\label{sec:conclusion}

\pnote{We have presented a declarative domain-specific language called TxForest
with provable guarantees.}
We have presented the design, syntax, and semantics of \TXFFull{}, a
domain-specific language for incrementally processing ad hoc data in concurrent
applications. \TXF{} aims to provide an easier and less error-prone approach to
modeling and interacting with a structured subset of a file system, which we
call a filestore. Specifically, the design provides an abstraction that handles
concurrent applications at low effort to the user. Additionally, we handle large
filestores by being automatically incremental.

\pnote{Zippers + local and global syntax and semantics to prove these guarantees.}
We achieve this by leveraging Huet's
Zippers~\cite{Huet:1997:ZIP:969867.969872} as our core abstraction.
Their traversal-based structure naturally lends itself to
incrementality and a simple, efficient logging scheme we use for our
optimistic concurrency control. We provide a core language with a
formal syntax and semantics based on zipper traversal, both for local,
single-threaded applications, and for a global view with
arbitrarily many Forest processes. We prove that this global view
enforces serializability between threads, that is, the resulting
effect on the file system of any set of concurrent threads is the same
as if they had run in some serial order.

\pnote{We have a prototype implementation.}
Our OCaml prototype provides a surface language mirroring Classic
Forest~\cite{forest-icfp:fisher+} which compiles down to the core language
mentioned above, and a library of functions for manipulating the filestore.
Additionally, we have built a simple course management system using this
prototype.

\bibliographystyle{splncs04}
\bibliography{main}

\begin{thebibliography}{10}
\providecommand{\url}[1]{\texttt{#1}}
\providecommand{\urlprefix}{URL }
\providecommand{\doi}[1]{https://doi.org/#1}

\bibitem{dilorenzo2016incremental}
DiLorenzo, J., Zhang, R., Menzies, E., Fisher, K., Foster, N.: Incremental
  {F}orest: a {DSL} for efficiently managing filestores. OOPSLA (2016)

\bibitem{Escriva:2016:DIW:2930611.2930642}
Escriva, R., Sirer, E.G.: The {D}esign and {I}mplementation of the {W}arp
  {T}ransactional {F}ilesystem. NSDI (2016)

\bibitem{forest-icfp:fisher+}
Fisher, K., Foster, N., Walker, D., Zhu, K.Q.: Forest: A language and toolkit
  for programming with filestores. ICFP (2011)

\bibitem{fisher+:pads}
Fisher, K., Gruber, R.: {PADS}: A domain specific language for processing ad
  hoc data. PLDI (2005)

\bibitem{focal-toplas}
Foster, J.N., Greenwald, M.B., Moore, J.T., Pierce, B.C., Schmitt, A.:
  {Combinators for Bidirectional Tree Transformations: A Linguistic Approach to
  the View Update Problem}. TOPLAS  (2007), short version in POPL '05.

\bibitem{Garcia:1998:PFT:319195.319224}
Garcia, J., Ferreira, P., Guedes, P.: The {PerDiS FS: A Transactional File
  System for a Distributed Persistent Store}. EW 8 (1998)

\bibitem{Huet:1997:ZIP:969867.969872}
Huet, G.: The {Z}ipper. J. Funct. Program.  (1997)

\bibitem{OlegZFSTalk}
Kiselyov, O.: Tool demonstration: {A} zipper based file/operating system.
  Haskell (2005)

\bibitem{Liskov:2004:TFS:1133572.1133592}
Liskov, B., Rodrigues, R.: {Transactional File Systems Can Be Fast}. EW 11
  (2004)

\bibitem{Moore:2008:HSO:1328438.1328448}
Moore, K.F., Grossman, D.: {High-level Small-step Operational Semantics for
  Transactions}. POPL (2008)

\bibitem{Schmuck:1991:ETQ:121132.121171}
Schmuck, F., Wylie, J.: {Experience with Transactions in QuickSilver}. SOSP
  (1991)

\end{thebibliography}

\iftr
\newpage

\appendix
\label{appendix}

\section{Proofs}\label{app:proofs}

\jdl{TODO: Add some prose in between theorems}
\jdl{TODO: And/or add a quick description of the theorem at the top}

\subsection{Serializability}\label{app:serializability}

\restatetheorem{thm:serializability}
\begin{proof}
  The definition of \TxFInitialName{} is in \figref{fig:more-tx-defs}.
  By the premise, \thmref{thm:serial-inductive} and $\msdiff{\TxBagMeta}{\SetEmpty}
  = \TxBagMeta$, we have:
  \begin{gather*}
    \exists \List{\TransactionMeta_1;\dots;\TransactionMeta_\IntMeta} \in Perm(\TxBagMeta). \\
    \FMStepTo{\FContMan{\TxFRestartF{\GFSMeta}{\TransactionMeta_1}}}{\LogMeta_1}{\FContNumGS{1}} \\
    \vdots \\
    ~ \band ~
      \FMStepTo{\FContMan{\TxFRestartF{\GFSMeta_{\IntMeta - 1}}{\TransactionMeta_\IntMeta{}}}}
        {\LogMeta_\IntMeta}
        {\FContNumFS{\IntMeta{}}{\GFSMeta'}}
  \end{gather*}

  By \lemref{lem:op-to-dem} and $\TxFInitialF{\ignoreArg}{\TransactionMeta}$, we have:
  \\
  $\FMStepToL{\FContMan{\TxFRestartF{\FSMeta}{\TransactionMeta}}}{\FContNFS{\FSMeta'}}$
  $\implies \FGSDenD = \FSMeta'$

  Thus, we have:
  \begin{gather*}
    \exists \List{\TransactionMeta_1;\dots;\TransactionMeta_\IntMeta} = \ListMeta \in Perm(\TxBagMeta). \\
    \FGSDenF{\GFSMeta}{\TransactionMeta_1} = \GFSMeta_1 ~ \band ~ \dots
    ~ \band ~ \FGSDenF{\GFSMeta_{\IntMeta - 1}}{\TransactionMeta_\IntMeta} = \GFSMeta' \\
    \implies \FGDenF{\ListMeta}{\GFSMeta} = \GFSMeta' \text{ (By definition of $\FGDenName$)}
  \end{gather*}

\end{proof}

\begin{lemma}[Operational to Denotational]
  \label{lem:op-to-dem}
  Let $\FSMeta$ be a file system and $\TransactionMeta =
  \TxMkF{\ThreadFS{\FSMeta''}}{\ignoreArg}$ be a transaction, then:
  \\
  $\FMAToACD{}{'}{\CommandMeta}{\ComSkip} \implies
    \FGSDenD = \FSMeta'$
\end{lemma}
\begin{proof}
  By rule induction from a similar big-step semantics and the equivalence of the
  small-step and big-step semantics.
  \jdl{TODO: Fix}

\end{proof}

\begin{definition}[Well-formed Transactions]
  \label{def:wf-trans}
  A transaction $\TransactionMeta$ is well-formed with respect to a file system
  $\GFSMeta$ and a global log $\GLogMeta$ (denoted $\TxWFT$) iff
  $\TransactionMeta$ comes from running an initial transaction for some number
  of steps and $\GFSMeta$ comes from merging the initial local file system of
  $\TransactionMeta$ with the more recent parts of $\GLogMeta$: \\
  $\TxWFT \iff$ \\
  $\exists \FSMeta, \TransactionMeta'. ~
  \TxFInitialF{\TransactionMeta'}{\FSMeta}
  ~ \band ~ \TxMSTToST{\TransactionMeta'}{\TransactionMeta}$

  $\band ~ \GLogMeta' =
  \Set{\LogEntryMeta \mid \bandf{\LogEntryMeta \in \GLogMeta}
  {\TxGetTSF{\LogEntryMeta} ~ \geq ~ \TxGetTSD}}$
  $~ \band ~ \GFSMeta = \TxFMergeF{\FSMeta}{\GLogMeta'}$
  \\ \\
  We call a thread pool, $\TxBagMeta$, well-formed in a similar manner when
  every transaction in it is well-formed: \\
  $\TxWFTP \iff \forall \TransactionMeta \in \TxBagMeta. ~ \TxWFT$
\end{definition}

\begin{definition}[Thread Pool Difference]
  \label{def:tp-diff}
  Difference on thread pools, written $\msdiff{\TxBagMeta}{\TxBagMeta'}$, is
  defined when there is a file system and global log for which both thread pools
  are well-formed. Then, thread pool difference is exactly normal multiset
  difference where equality on elements is defined by having the same initial
  transaction, $\TransactionMeta'$, as seen in \defref{def:wf-trans}.
\end{definition}

\begin{theorem}[Inductive Serializability]
  \label{thm:serial-inductive}
  Let $\GFSMeta, \GFSMeta'$ be file systems, $\GLogMeta,\GLogMeta'$ be global
  logs, and $\TxBagMeta, \TxBagMeta'$ thread pools such that $\TxWFTP$, then: \\
  \\
  $\TxMStepTo{\TxContAll{}}{\TxContAll{'}} \implies$ \\
  $ \exists \List{\TransactionMeta_1;\dots;\TransactionMeta_\IntMeta} \in
  Perm(\msdiff{\TxBagMeta}{\TxBagMeta'}).$ \\
  $ \FMStepTo{\FContMan{\TxFRestartF{\GFSMeta}{\TransactionMeta_1}}}{\LogMeta_1}{\FContNumGS{1}} \\
    \vdots \\
    ~ \band ~
      \FMStepTo{\FContMan{\TxFRestartF{\GFSMeta_{\IntMeta - 1}}{\TransactionMeta_\IntMeta{}}}}
        {\LogMeta_\IntMeta}
        {\FContNumFS{\IntMeta{}}{\GFSMeta'}}
  $

\end{theorem}
\begin{proof}
  By induction on the multi-step relation $\TxMStep$. See \figref{fig:more-tx-defs}
  for a definition of \TxFRestartName{}.
  The reflexive case is straight-forward, while the transitive step relies on
  \lemref{lem:wf-preservation} to be able to apply the inductive hypothesis
  twice.
  The single-step case is significantly more complicated and entirely
  covered in \lemref{lem:ss-serial}.

\end{proof}

\begin{lemma}[Well-formedness Preservation]
  \label{lem:wf-preservation}
  Let $\GFSMeta, \GFSMeta'$ be file systems, $\GLogMeta,\GLogMeta'$ be global
  logs, and $\TxBagMeta, \TxBagMeta'$ thread pools, then: \\
  $\TxWFTP ~ \band ~ \TxMAToA{}{'}
  \implies \TxWFF{\GFSMeta', \GLogMeta'}{\TxBagMeta'}$
\end{lemma}
\begin{proof}
  By straightforward induction on the multi-step relation $\TxMStep$.

\end{proof}

\begin{lemma}[Single-step Serializability]
  \label{lem:ss-serial}
  Let $\GFSMeta, \GFSMeta'$ be file systems, $\GLogMeta,\GLogMeta'$ be global
  logs, and $\TxBagMeta, \TxBagMeta'$ thread pools such that $\TxWFTP$, then: \\
  \\
  $\TxStepTo{\TxContAll{}}{\TxContAll{'}}
    ~ \band ~ \TransactionMeta \in (\msdiff{\TxBagMeta}{\TxBagMeta'}) \implies$ \\
  $ \FMStepTo{\FContMan{\TxFRestartF{\GFSMeta}{\TransactionMeta}}}
      {\LogMeta}{\FContAllFC{'}{\GFSMeta'}{\ComSkip}}$
\end{lemma}
\begin{proof}
  By induction on the single-step relation $\TxStep$. The theorem holds
  vacuously unless the step is a commit. If the step is a commit, then
  it follows from \thmref{thm:merge-prop}.

\end{proof}

\begin{theorem}[Merge Property]
  \label{thm:merge-prop}
  Let $\TransactionMeta = \TxD$ be a transaction,
  $\GLogMeta$ a global log,
  and $\GFSMeta$ a file system such that
  $\TxWFT$, $\TxFCheckLogD$, and $\TxFMergeD = \MergeFS$. \\
  Then,
  $ \FMStepTo{\FContMan{\TxFRestartF{\GFSMeta}{\TransactionMeta}}}
      {\LogMeta}{\FContAllFC{'}{\MergeFS}{\ComSkip}}$
\end{theorem}

\begin{proof}
  Follows directly from \lemref{lem:merge}. We use $\TxWFT$ to conclude
  \\
  $\exists \TransactionMeta'. ~
    \TxMSTToST{\TransactionMeta'}{\TransactionMeta}$ and
  that $\TxWFF{\GFSMeta, \GLogMeta}{\TransactionMeta'}$.
  \\
  Then we can apply the lemma.

\end{proof}

\begin{lemma}[Partial Check Log]
  \label{lem:cl-partial}
  \\
  $\TxFCheckLogL{(\ListAppendLF{\LogMeta}{\LogMeta'})}
  \implies \bandf{\TxFCheckLogD}{\TxFCheckLogL{\LogMeta'}}$
\end{lemma}
\begin{proof}
  $\TxFExtractPathsF{(\ListAppendLF{\LogMeta}{\LogMeta'})} =
    \TxFExtractPathsF{\LogMeta} ~ \union ~
    \TxFExtractPathsF{\LogMeta'}$
\end{proof}

\begin{lemma}[Check Log Property]
  \label{lem:checkLog}
  Let $\TransactionMeta$ be a transaction with log and timestamp, $\LogMeta$ and
  $\TimestampMeta$, respectively, $\GLogMeta$ a global log, and $\GFSMeta$ a
  file system such that $\TxWFT$. \\
  Then, $\TxFCheckLogD \implies \TxFCompatD$
\end{lemma}
\begin{proof}
  By induction on the structure of the log, $\LogMeta$ and in the non-empty
  case, induction on the multi-step function $\TxMStep$ having used $\TxWFT$ to
  establish a derivation. See \figref{fig:log-compat-def} for a definition of $\sim$.

  Uses \lemref{lem:real-reads} and \lemref{lem:g-to-l}
  as well as \lemref{lem:log-com-combo} (in the transitive case).
  \begin{lemma}[Reads Correspond to Values]
    \label{lem:real-reads}

    $\TxLReadD \in \TxFCanonizeD
    ~ \band ~ \FMAToA{}{\LogMeta}{'}$

    $\implies \MapGet{\FSMeta}{\PathMeta} = \ContentsMeta$
  \end{lemma}
  \begin{proof}
    By induction on the multi-step function $\FMStep{\LogMeta}$.

    Uses \lemref{lem:real-reads-aux} in transitive case.
    \begin{lemma}
      \label{lem:real-reads-aux}

      $\forall \PathMeta. \not \exists \PathMeta'. ~
      \bandf{\TxFSubPathD}{\PathMeta' \in \TxFWritesD}$ \\
      $~ \band ~ \FMAToA{}{\LogMeta}{'}
        \implies \MapGet{\FSMeta}{\PathMeta} = \MapGet{\FSMeta'}{\PathMeta}$
    \end{lemma}
    \begin{proof}
      By induction on the multi-step function $\FMStep{\LogMeta}$.

    \end{proof}

  \end{proof}
  \begin{lemma}[Global to Local]
    \label{lem:g-to-l}
    Let $\GFSMeta$ be a file system and $\GLogMeta$ be a global
    log. Then, \\
    $\TxMStepTo{\TxContST{\TxD}}
      {\\\TxContST{\TxMkF{\ThreadP{'}}{\TxSL{\ListAppendLF{\LogMeta}{\LogMeta'}}}}}$ \\
    $\band ~ \TxFCheckLogL{(\ListAppendLF{\LogMeta}{\LogMeta'})} \implies$ \\
    $ \FMAToACL{}{'}{\CommandMeta}{\CommandMeta'}{\LogMeta'}$
  \end{lemma}
  \begin{proof}
    By induction on the multi-step relation $\TxMStep$. The transitive case relies
  on \lemref{lem:cl-partial}, but is straightforward. The reflexive case is trivial.
  In the step case, commit and abort are ruled out, leaving the single thread step.
  \end{proof}

\end{proof}

\begin{lemma}[Merge Lemma]
  \label{lem:merge}
  Let $\TransactionMeta = \TxThMDD$ and
      $\TransactionMeta' = \TxMkF{\ThreadMeta'}{\TxSL{\ListAppendLF{\LogMeta}{\LogMeta'}}}$
  be transactions, $\GLogMeta$ a global log, and
      $\GFSMeta$ a file system such that
      $\bandf{\TxWFT}{\TxFCheckLogL{(\ListAppendLF{\LogMeta}{\LogMeta'})}}$. \\
  Then, \\
  $\TxMSTToST{\TransactionMeta}{\TransactionMeta'}
    ~ \band ~ \TxFMergeD = \GFSMeta'
    ~ \band ~ \TxFMergeF{\GFSMeta'}{\LogMeta'} = \MergeFS$
  \\
  $\implies
    \FMStepTo{\FContMan{\TxFInsertF{\GFSMeta'}{\ThreadMeta}}}
      {\LogMeta'}
      {\FContMan{\TxFInsertF{\MergeFS}{\ThreadMeta'}}}$
\end{lemma}
\begin{proof}
  By induction on the multi-step relation $\TxMStep$. The transitive case relies
  on \lemref{lem:cl-partial} and \lemref{lem:wf-preservation} to get
  intermediate $\TxFCheckLogName$ and well-formedness results. Additionally, it
  relies on the fact that the global log monotonically grows at the same time as
  the file system changes, which means that the intermediate steps have the same
  $\GFSMeta$ and $\GLogMeta$.

  The single-step case first rules out commit (because a transaction remains)
  and restart (because $\TxFCheckLogD$). With only the local step case
  remaining, we induct on the single-step relation $\FStep{\LogMeta'}$.

  The IMP rules are straightforward, and mostly do not affect the file system.
  In the Forest Command case, we use \lemref{lem:inter-wf} and
  \lemref{lem:checkLog} to derive $\TxFCompatF{\GFSMeta'}{\LogMeta'}$,
  then use \lemref{lem:merge-fn} and \lemref{lem:merge-fu} for Forest
  Navigations and Forest Updates respectively. Since Forest Navigations only
  produce reads (by \lemref{lem:merge-fn}), we also note that $\MergeFS =
  \GFSMeta'$ in these cases.
\end{proof}

\begin{lemma}[Intermediate Well-formedness]
  \label{lem:inter-wf}
  Let $\TransactionMeta = \TxThMDD$ and
  $\TransactionMeta' = \TxMkF{\ThreadMeta'}{\TxSL{\ListAppendLF{\LogMeta}{\LogMeta'}}}$
  be transactions, $\GLogMeta$ a global log, and
  $\GFSMeta$ a file system. \\
  Then, \\
  $\TxMSTToST{\TransactionMeta}{\TransactionMeta'}
  ~ \band ~ \TxWFT
  ~ \band ~ \TxFMergeD = \MergeFS
  \implies \TxWFF{\MergeFS,\ListAppendLF{\GLogMeta}{(\TxAddTSF{\TimestampMeta}{\LogMeta})}}{\TransactionMeta'}$

\end{lemma}
\begin{proof}
  Follows from \defref{def:wf-trans}.
\end{proof}

\begin{lemma}[Merge: Forest Updates]
  \label{lem:merge-fu}
  If $\TxFCompatD$, then \\
  $\forall \FUpdateMeta. ~ (\exists \FSMeta. ~
  \FCDenFS{\FUpdateMeta}{\FSMeta} = (\ContextMeta,\pair{\PathSetMeta}{\FSMeta'},\LogMeta)
  ~ \band ~ \TxFMergeD = \MergeFS$ \\
  \text{\quad \quad}
  $\implies \FCDenFS{\FUpdateMeta}{\GFSMeta} = (\ContextMeta,\pair{\PathSetMeta}{\MergeFS},\LogMeta)$
\end{lemma}
\begin{proof}
  By induction on the Forest Updates $\FUpdateMeta$. Uses \lemref{lem:merge-exp}
  for subexpressions and \lemref{lem:log-com-reads} and
  \lemref{lem:log-com-part} to focus on the write
  portions of the log.
\jdl{TODO: Maybe show an example case?}
\end{proof}

\begin{lemma}[Merge: Forest Navigations]
  \label{lem:merge-fn}
  If $\TxFCompatD$, then \\
  $\forall \FNavigationMeta. ~ (\exists \FSMeta. ~
  \FCDenFS{\FNavigationMeta}{\FSMeta} = (\ContextMeta,\GContextMeta,\LogMeta)$
  \\
  \text{\quad \quad}
  $\implies \TxFReadsD = \LogMeta
    ~ \band ~ \FCDenFS{\FNavigationMeta}{\GFSMeta} = (\ContextMeta,\GContextMeta,\LogMeta)$
\end{lemma}
\begin{proof}
  By induction on the Forest Navigations $\FNavigationMeta$ and mutually
  dependent on \lemref{lem:merge-exp}. Uses \lemref{lem:log-com-reads} and
  \lemref{lem:log-com-part} to be able to apply \lemref{lem:merge-exp} and the
  induction hypothesis in sequence.

\end{proof}

\begin{lemma}[Merge: Expressions]
  \label{lem:merge-exp}
  If $\TxFCompatD$, then \\
  $\forall \ExpMeta. ~ (\exists \FSMeta. ~ \FEDenDE{\ExpMeta} = \pair{\ValMeta}{\LogMeta}$ \\
  \text{\quad \quad}
  $\implies \TxFReadsD = \LogMeta
    ~ \band ~ \FEDenFS{\ExpMeta}{\GFSMeta} = \pair{\ValMeta}{\LogMeta})$
\end{lemma}
\begin{proof}
  By induction on the expressions $\ExpMeta$ and mutually dependent on
  \lemref{lem:merge-fn}.
  For $\FFunVerify$, there is a further induction on $\SpecMeta$.
  For $\FExpRunED$, we apply the induction hypothesis
  twice and for $\FExpRunCD$, we apply \lemref{lem:merge-fn} once and the
  induction hypothesis once. In both cases, we rely on
  \lemref{lem:log-com-reads} and \lemref{lem:log-com-part}.

\end{proof}

\begin{lemma}[Log Compatibility with Reads]
  \label{lem:log-com-reads}

  $\TxFCompatF{\GFSMeta}{(\ListAppendLF{\LogMeta}{\LogMeta'})}
  ~\band~ \TxFReadsF{\LogMeta} = \LogMeta \implies
  \bandf{\TxFCompatD}{\TxFCompatF{\GFSMeta}{\LogMeta'}}
  $
\end{lemma}
\begin{proof}
  The first part follows from \lemref{lem:log-com-part}. The second
  from the fact that
  $\TxFCanonizeF{\LogMeta'}
  \subseteq
  \TxFCanonizeF{(\ListAppendLF{(\TxFReadsF{\LogMeta})}{\LogMeta'})}$.
\end{proof}

\begin{lemma}[Log Compatibility Combination]
  \label{lem:log-com-combo}

  $\bandf{\TxFCompatD}{\TxFCompatF{\GFSMeta}{\LogMeta'}}
  \implies \TxFCompatF{\GFSMeta}{(\ListAppendLF{\LogMeta}{\LogMeta'})}$
\end{lemma}
\begin{proof}
  $\TxFReadsF{(\TxFCanonizeF{(\ListAppendLF{\LogMeta}{\LogMeta'})})} \subseteq
  \TxFReadsF{(\TxFCanonizeD)} ~ \union ~ \TxFReadsF{(\TxFCanonizeF{\LogMeta'})} $
\end{proof}

\begin{lemma}[Log Compatibility Parts]
  \label{lem:log-com-part}

  $\TxFCompatF{\GFSMeta}{(\ListAppendLF{\LogMeta}{\LogMeta'})}
    \implies \TxFCompatD$
\end{lemma}
\begin{proof}
  $\TxFReadsF{(\TxFCanonizeD)} \subseteq \TxFReadsF{(\TxFCanonizeF{(\ListAppendLF{\LogMeta}{\LogMeta'})})}$
\end{proof}

\begin{figure}
\begin{code}
  \TxFCompatD \eqdef \(\forall \TxLReadD\in\TxFCanonizeD\). \MapGet{\GFSMeta}{\PathMeta} = \ContentsMeta
  \TxFCanonizeD \eqdef \foldFunF{\ListEmpty}{\LogMeta}{\TxFNecessaryName}
  \codeskip{}
  \TxFNecessaryD \eqdef
    if \TxFSubPathF{\PathMeta}{(\TxFWritesF{acc})} \bor \(\PathMeta\in\TxFReadsF{acc}\)
    then \(acc\)
    else \ListAppendIF{(\TxLReadD)}{acc}

  \TxFNecessaryL{(\TxLWriteFileD)} \eqdef
    \ListAppendIF{(\TxLWritePathD)}{(\TxFNecessaryDW)}
  \TxFNecessaryL{(\TxLWriteDirF{(\FSFileD)}{(\FSDirD)}{\PathMeta})} \eqdef
    \ListAppendIF{(\TxLWritePathD)}{(\TxFNecessaryF{acc}{(\TxLReadF{(\FSFileD)}{\PathMeta})})}
  \TxFNecessaryL{(\TxLWriteDirF{(\FSDirF{\SetMetaT{'}})}{(\FSDirD)}{\PathMeta})} \eqdef
    \foldFunF{(\TxFNecessaryF{acc}{(\TxLReadF{(\FSDirF{\SetMetaT{'}})}{\PathMeta})})}{}{}
      (\((\SetMeta\setminus\SetMetaT{'})\union(\SetMetaT{'}\setminus\SetMeta)\))
      (\anonFunF{acc \StrMeta}{\ListAppendIF{\TxLWritePathF{\PathConsS}}{acc}})
\end{code}
\caption{Log Compatibility definition}
\label{fig:log-compat-def}
\end{figure}

\begin{figure}
  \begin{code}
    \TxFRestartF{\GFSMeta}{\TxD} \eqdef \ThreadStartF{\GFSMeta}{\text{'}}
      where \pair{\PathMetaT{'}}{\AddT{\ZipperMeta}{'}} = \FFGoToRootD
    \codeskip{}
    \TxFInitialF{\FSMeta}{\TxStartLD} when \OptIsNoneF{\ZipGetParent} \eqdef \btrue
    \TxFInitialF{\textIgnoreArg}{\textIgnoreArg} \eqdef \bfalse
  \end{code}

  \caption{Definitions of Initial and restart}
  \label{fig:more-tx-defs}
\end{figure}

\subsection{Properties}\label{app:properties}

\paragraph*{Consistency.} We restate the consistency theorems from
\secref{sec:forest} and give the main idea of their proofs.

\restatetheorem{thm:const-to-pconst}
\restatetheorem{thm:pconst-monotonicity}
\restatetheorem{thm:pconst-eq-const}
\begin{proof}
  The proofs of these three theorems are straightforward by induction on the
  structure of the specification in $\ZipperMeta$. The theorems ignore the log
  portion of partial consistency. \FPathCoverName{} generates a path set
  by traversing the whole filestore until it reaches a fixed point.
\end{proof}

\paragraph*{Core Calculus Equivalences.} We present several equivalences in the
core calculus.

\begin{definition}[Equivalence modulo logs]
  \label{def:equiv}
  We define equivalence modulo logs inductively as follows:
  \[
    \begin{aligned}
    \bot &\undefeq \ignoreArg
    \\
    \ignoreArg &\undefeq \bot
    \\
    (\ContextExp,\GContextExp,\ignoreArg) &\undefeq (\ContextExp,\GContextExp,\ignoreArg)
    \\
    \FunctionMeta &\undefeq \FunctionMeta' & \text{ when } & \forall \ValMeta. ~ \FunctionMeta ~ \ValMeta \undefeq \FunctionMeta' ~ \ValMeta
    \end{aligned}
  \]
\end{definition}

\begin{lemma}[Core Calculus Equivalences]
  \label{lem:equivalences}
  \[
  \begin{aligned}
      \FCDen{\ComSeqF{\FComDown}{\FComUp}} & \undefeq \FCDen{\ComSkip} \\
      \FCDen{\ComSeqF{\FComIntoOpt}{\FComOut}} & \undefeq \FCDen{\ComSkip} \\
      \FCDen{\ComSeqF{\FComIntoComp}{\FComOut}} & \undefeq \FCDen{\ComSkip} \\
      \FCDen{\ComSeqF{\FComIntoPair}{\FComOut}} & \undefeq \FCDen{\ComSkip} \\
      \FCDen{\ComSeqF{\FComNext}{\FComPrev}} & \undefeq \FCDen{\ComSkip} \\
      \FCDen{\ComSeqF{\FComPrev}{\FComNext}} & \undefeq \FCDen{\ComSkip} \\
  \end{aligned}
  \]
\end{lemma}

\paragraph*{Round-Tripping Laws.} We present several round-tripping laws in the style of
lenses~\cite{focal-toplas}.

\begin{lemma}[Round-Tripping Laws]
  \label{lem:rt-laws}
  \[
  \begin{aligned}
      \FCDen{\FComStoreFileF{\FExpFetchFile}} & \undefeq
        \FCDen{\ComSkip} & \text{File-Load-Store} \\
      \FCDen{\FComStoreDirF{\FExpFetchDir}} & \undefeq
        \FCDen{\ComSkip} & \text{Dir-Load-Store} \\
      \FCDen{\ComSeqF{\FComStoreFileF{\StrMeta_1}}{\FComStoreFileF{\StrMeta_2}}} & \undefeq
        \FCDen{\FComStoreFileF{\StrMeta_2}} & \text{File-Store-Store} \\
      \FCDen{\ComSeqF{\FComCreatePath}{\FComCreatePath}} & \undefeq
        \FCDen{\FComCreatePath} & \text{CreatePath-Store-Store} \\
      \MapGet{(project\_env ~ (\FCDenDCG{\ComSeqF{\FComStoreFileF{\StrMeta}}{\ComAssignF{\VarMeta}{\FExpFetchFile}}}))}
      {\VarMeta} = \StrMeta
      \span &  \text{File-Store-Load} \\
      \MapGet{(project\_env ~ (\FCDenDCG{\ComSeqF{\FComStoreDirF{\SetMeta}}{\ComAssignF{\VarMeta}{\FExpFetchDir}}}))}
      {\VarMeta} = \SetMeta
      \span &  \text{Dir-Store-Load}
  \end{aligned}
  \]
\end{lemma}

The store-load laws also require the denotation to be defined in order to be meaningful.
Note that
$\FCDen{\ComSeqF{\FComStoreDirF{\SetMeta_1}}{\FComStoreDirF{\SetMeta_2}}}
\undefeq \FCDen{\FComStoreDirF{\SetMeta_2}}$ is conspicuously missing. In fact,
it does not hold. Consider the situation where $\SetMeta_2$ is the current
contents of the given directory. In this case, $\FComStoreDirF{\SetMeta_2}$ is a
no-op, and thus the right-hand side is equivalent to $\ComSkip$. However, if,
for example, $\SetMeta_1 = \ListEmpty$, then the left-hand side will turn every
child into $\FSFileEmpty$.

\section{Additional definitions}

\begin{figure}
  \fbox{
  $
  \begin{aligned}
    \FFMkFileName &: \fun{\FSSet}{\fun{\PathSet}{\fun{\StrSet}{\product{\FSSet}{\LogSet}}}}
    \\
    \FFMkNodeName &: \fun{\FSSet}{\fun{\PathSet}{\fun{2^{\StrSet}}{\product{\FSSet}{\LogSet}}}}
    \\
    \FFAddNodeName &: \fun{\FSSet}{\fun{\PathSet}{\fun{\StrSet}{\product{\FSSet}{\LogSet}}}}
    \\
    \TxFCloseName &: \fun{\FSSet}{\FSSet}
  \end{aligned}
  $
  }
  \begin{code}
    \FFMkFileD \eqdef
      \ExpLet{\pair{\FSMetaT{'}}{\LogMetaT{'}}}{\FFAddNodeS}
      \ExpLet{\LogMeta}{\ListAppendLF{\LogMetaT{'}}{(\TxLWriteFileF{\MapGet{\FSMetaT{'}}{\PathMeta}}{(\FSFileD)}{\PathMeta})}}
      \pair{\TxFCloseF{(\FSMetaT{'}\subst{\PathMeta}{\FSFileD})}}{\LogMeta}

    \codeskip{}

    \FFMkNodeD \eqdef
      \ExpLet{\pair{\FSMetaT{'}}{\LogMetaT{'}}}{\FFAddNodeS}
      \ExpLet{\LogMeta}{\ListAppendLF{\LogMetaT{'}}{(\TxLWriteDirF{\MapGet{\FSMetaT{'}}{\PathMeta}}{(\FSDirD)}{\PathMeta})}}
      \pair{\TxFCloseF{(\FSMetaT{'}\subst{\PathMeta}{\FSDirD})}}{\LogMeta}

    \codeskip{}

    \FFAddNodeD \eqdef
      \ExpMatchStart{\MapGet{\FSMeta}{\PathMeta}}
      \ExpMatchCase{\bot}
        \ExpLet{\pair{\FSMetaT{'}}{\LogMetaT{'}}}{\FFAddNodeS}
        \ExpLet{\LogMeta}{\ListAppendLF{\LogMetaT{'}}{(\TxLWriteDirF{(\FSFileEmpty)}{(\FSDirS)}{\PathConsS})}}
        \pair{\TxFCloseF{(\FSMetaT{'}\subst{\PathMeta}{\FSDirS})}}{\LogMeta}
      \ExpMatchCase{\FSFileF{\StrMetaT{'}}}
        \pair{\TxFCloseF{(\FSMeta\subst{\PathMeta}{\FSDirS})}}{\List{\TxLWriteDirF{(\FSFileF{\StrMetaT{'}})}{(\FSDirS)}{\PathConsS}}}
      \ExpMatchCaseWhen{\FSDirD}{\StrMeta \notin \SetMeta}
        \pair{\TxFCloseF{(\FSMeta\subst{\PathMeta}{\FSDirF{(\SetMeta \union \TextSet{\StrMeta})}})}}{\List{\TxLWriteDirF{(\FSDirD)}{(\FSDirF{(\SetMeta \union \TextSet{\StrMeta})})}{\PathConsS}}}
      \ExpMatchCaseWhen{\FSDirD}{\StrMeta \in \SetMeta} \pair{\FSMeta}{\ListEmpty}

    \codeskip{}

    \TxFCloseD \eqdef close_at \FSMeta \PathRoot
      where close_at \FSMeta \PathMeta \eqdef
        \ExpMatchStart{\MapGet{\FSMeta}{\PathMeta}}
        \ExpMatchCase{\FSDirD}
          \ExpLet{\FSMetaT{'}}{\foldFunF{\FSMeta}{\SetMeta}{\text{close\_at}}}
          \ExpLet{\SetMetaT{'}}{\text{\Set{\(\PathMetaT{'}\in\FSMeta\mid\TxFSubPathF{\PathMetaT{'}}{\PathMeta} \band \forall\StrMeta\in\SetMeta. \neg\TxFSubPathF{\PathMetaT{'}}{\PathConsS}\)}}}
          \ExpLet{\FSMetaT{''}}{\foldFunF{\FSMetaT{'}}{\SetMetaT{'}}{(\anonFunF{\FSMeta \PathMetaT{'}}{\FSMeta\subst{\PathMeta}{\bot}})}}
          \FSMetaT{''}\subst{\PathMeta}{\FSDirD}
        \ExpMatchCase{\FSFileD}
          \ExpLet{\FSMetaT{'}}{\foldFunF{\FSMeta}{\text{\Set{\(\PathMetaT{'}\in\FSMeta\mid\TxFSubPathF{\PathMetaT{'}}{\PathMeta}\)}}}{(\anonFunF{\FSMeta \PathMetaT{'}}{\FSMeta\subst{\PathMeta}{\bot}})}}
          \FSMetaT{'}\subst{\PathMeta}{\FSFileD}
        \ExpMatchCase{\bot} \FSMeta\subst{\PathMeta}{\FSFileEmpty}
  \end{code}
  \caption{Helper Functions}
  \label{fig:forest-helpers}
\end{figure}

\else
\fi
\end{document}